\documentclass[10pt]{llncs}
\usepackage[T1]{fontenc}
\usepackage[latin9]{inputenc}
\usepackage{color}
\usepackage{float}
\usepackage{amsmath}
\usepackage{amssymb}
\usepackage{graphicx}

\makeatletter

\floatstyle{ruled}
\newfloat{algorithm}{tbp}{loa}
\providecommand{\algorithmname}{Algorithm}
\floatname{algorithm}{\protect\algorithmname}

\usepackage{llncsdoc}\@ifundefined{definecolor}
 {\usepackage{color}}{}
\usepackage{float}\usepackage{amsfonts}\usepackage{multicol}

\newtheorem{rem}{Remark}

\makeatother

\begin{document}

\title{Constrained Fault-Tolerant Resource Allocation%
\thanks{Part of section 5 of this paper has appeared in \cite{kewen2011cocoon}.%
}}

\author{\author{Kewen Liao$^1$, Hong Shen$^{1,2}$ and Longkun Guo$^{2}$ }}

\institute{\institute{$^1$School of Computer Science \\ The University of Adelaide, Adelaide, Australia  \\  $^2$School of Computer and Information Technology\\ Sun Yat-sen University, Guangzhou, China \\ \email{\{kewen, hong\}@cs.adelaide.edu.au} \\ \email{longkun.guo@gmail.com}}}
\maketitle
\begin{abstract}
In the Constrained Fault-Tolerant Resource Allocation ($FTRA$) problem,
we are given a set of sites containing facilities as resources, and
a set of clients accessing these resources. Specifically, each site
$i$ is allowed to open at most $R_{i}$ facilities with cost $f_{i}$
for each opened facility. Each client $j$ requires an allocation
of $r_{j}$ open facilities and connecting $j$ to any facility at
site $i$ incurs a connection cost $c_{ij}$. The goal is to minimize
the total cost of this resource allocation scenario.

\quad{}\enskip{}$FTRA$ generalizes the Unconstrained Fault-Tolerant
Resource Allocation ($FTRA_{\infty}$) \cite{kewen2011cocoon} and
the classical Fault-Tolerant Facility Location ($FTFL$) \cite{Jain00FTFL}
problems: for every site $i$, $FTRA_{\infty}$ does not have the
constraint $R_{i}$, whereas $FTFL$ sets $R_{i}=1$. These problems
are said to be uniform if all $r_{j}$'s are the same, and general
otherwise.

\quad{}\enskip{}For the general metric $FTRA$, we first give an
LP-rounding algorithm achieving an approximation ratio of 4. Then
we show the problem reduces to $FTFL$, implying the ratio of 1.7245
from \cite{JaroslawFTFL1.725}. For the uniform $FTRA$, we provide
a 1.52-approximation primal-dual algorithm in \textcolor{black}{$O\left(n^{4}\right)$
time, where $n$ is the total number of sites and clients.} We also
consider the Constrained Fault-Tolerant $k$-Resource Allocation ($k$-$FTRA$)
problem where additionally the total number of facilities can be opened
across all sites is bounded by $k$. For the uniform $k$-$FTRA$,
we give the first constant-factor approximation algorithm with a factor
of 4. Note that the above results carry over to $FTRA_{\infty}$ and
$k$-$FTRA_{\infty}$.
\end{abstract}
\markboth{Kewen Liao, Hong Shen and Longkun Guo}{Constrained Fault-Tolerant
Resource Allocation }\thispagestyle{empty}

\section{Introduction}

In the \textit{Constrained Fault-Tolerant Resource Allocation} ($FTRA$)
problem introduced in \cite{kewen2011cocoon}, we are given a set
$\mathcal{F}$ of sites and a set \textcolor{black}{$\mathcal{C}$}
of clients, where $\left|\mathcal{F}\right|=n_{f}$, $\left|\mathcal{C}\right|=n_{c}$
and $n=n_{f}+n_{c}$. Each site $i\in\mathcal{F}$ contains at most
$R_{i}$ ($R_{i}\geq1$) facilities to open as resources and each
client $j\in\mathcal{C}$ is required to be allocated $r_{j}$ ($r_{j}\geq1$)
open facilities. Note that in $FTRA$, facilities at the same site
are different and $\max_{j\in\mathcal{C}}r_{j}\leq\sum_{i\in\mathcal{F}}R_{i}$.
Moreover, opening a facility at site $i$ incurs a cost $f_{i}$ and
connecting $j$ to any facility at $i$ costs $c_{ij}$. The objective
of the problem is to minimize the sum of facility opening and client
connection costs under the resource constraint $R_{i}$. This problem
is closely related to the \textit{Unconstrained Fault-Tolerant Resource
Allocation} ($FTRA_{\infty}$)%
\footnote{The problem was also called \textit{Fault-Tolerant Facility Allocation
}($FTFA$) in \cite{shihongftfa} and \textit{Fault-Tolerant Facility
Placement }($FTFP$) in \cite{yan2011approximation}. Nevertheless,
we reserve our names for identifying the different application-oriented
resource allocation scenarios. Our naming convention also follows
from \cite{fujito2005better,hua2009exact,kolliopoulos2003approximating}
for the set cover problems.%
} \cite{kewen2011cocoon}, the classical \textit{Fault-Tolerant Facility
Location }($FTFL$) \cite{Jain00FTFL} and\textit{ Uncapacitated Facility
Location }($UFL$) \cite{Shmoys97FL} problems. Both $FTRA_{\infty}$
and $FTFL$ are special cases of $FTRA$: $R_{i}$ is unbounded in
$FTRA_{\infty}$, whereas $\forall i\in\mathcal{F}:\, R_{i}=1$ in
$FTFL$. These problems are said to be \textit{uniform} if all $r_{j}$'s
are same, and\textit{ general} otherwise. If $\forall j\in\mathcal{C}:\, r_{j}=1$,
they all reduce to $UFL$. Fig. 1 displays an $FTRA$ instance with
a feasible solution. We notice that both $FTRA$ and $FTRA_{\infty}$
have potential applications in numerous distributed systems such as
cloud computing, content delivery networks, Web services provision
and etc. The fault-tolerance attribute ($r_{j}$) can be also viewed
as the parallel processing capability of these systems. Unless elsewhere
specified, we consider the problems in metric space, that is, the
connection costs $c_{ij}$'s satisfy the metric properties like the
triangle inequality and etc.\textcolor{black}{{} Note that even the
simplest non-metric $UFL$ is hard to approximate better than $O\left(\log n\right)$
unless $NP\subseteq DTIME\left[n^{O\left(\log\log n\right)}\right]$}
\cite{Sviridenko02improved1.58}.

\begin{figure}
\begin{centering}
\includegraphics[scale=0.7]{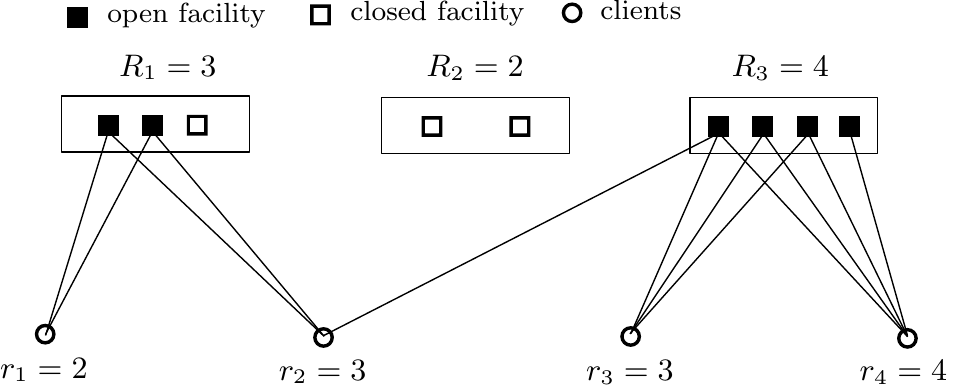}
\par\end{centering}

\caption{An $FTRA$ instance with a feasible solution}
\end{figure}

\textbf{Related Work.} Primal-dual and LP-rounding are two typical
approaches in designing approximation algorithms for the facility
location problems. Starting from the most basic and extensively studied
$UFL$ problem, there are JV \cite{jain01approximation}, MMSV \cite{Mohammad1.861}
and JMS \cite{Jain02greedy} primal-dual algorithms obtaining approximation
ratios of 3, 1.861 and 1.61 respectively. In addition, Charikar and
Guha \cite{Charikar051.7281.853} improved the result of the JV algorithm
to 1.853 and Mahdian et al. \cite{Mohammad06FLP} improved that of
the JMS algorithm to 1.52, both using the standard cost scaling and
greedy augmentation techniques. Shmoys et al. \cite{Shmoys97FL} first
gave a filtering based LP-rounding algorithm achieving the constant
ratio of 3.16. Following this, Guha and Khuller \cite{Guha99greedy}
improved the factor to 2.41 with greedy augmentation. Later, Chudak
and Shmoys \cite{Chudak0312e} came up with the clustered randomized
rounding algorithm which further reduces the ratio to 1.736. Based
on their algorithm, Sviridenko \cite{Sviridenko02improved1.58} applied
pipage rounding to obtain 1.582-approximation. Byrka \cite{jaroslaw2010optimal}
achieved the ratio of 1.5 using a bi-factor result of the JMS algorithm.
Recently, Li's more refined analysis in \cite{Li2011} obtained the
current best ratio of 1.488, which is close to the 1.463 lower bound
established by Guha and Khuller \cite{Guha99greedy}. For the \textcolor{black}{non-metric
$UFL$, there are two $O\left(\log n\right)$ approximation algorithms
\cite{Hochbaum-1982,Lin92filting} based on the greedy and LP-rounding
approaches respectively.}

Comparing to $UFL$, $FTFL$ seems more difficult to approximate.
For the general $FTFL$, the primal-dual algorithm in \cite{Jain00FTFL}
yields a non-constant factor \textcolor{black}{$O\left(\log n\right)$}.
Constant results exist only for the uniform case. In particular, Jain
et al. \cite{Jain03dualfitting,MohammadThesis2004} showed their MMSV
and JMS algorithms for $UFL$ can be adapted to the uniform $FTFL$
while preserving the ratios of 1.861 and 1.61 respectively. Swamy
and Shmoys \cite{Swamy08FTFL2.076} improved the result to 1.52. On
the other hand, LP-rounding approaches are more successful for the
general $FTFL$. Guha et al. \cite{Guha03FTFL2.41} obtained the first
constant factor algorithm with the ratio of 2.408. Later, this was
improved to 2.076 by Swamy and Shmoys \cite{Swamy08FTFL2.076} with
several rounding techniques. Recently, Byrka et al. \cite{JaroslawFTFL1.725}
used dependent rounding and laminar clustering techniques to get the
current best ratio of 1.7245.

Allowing parallel connections to multiple facilities at the same site,
$FTRA_{\infty}$ was first introduced by Xu and Shen \cite{shihongftfa}
and they claimed a 1.861 approximation algorithm which runs in pseudo-polynomial
time for the general case. Liao and Shen \cite{kewen2011cocoon} studied
the uniform case of the problem and presented a 1.52 approximation
algorithm using a star-greedy approach. The general case of the problem
was also studied by Yan and Chrobak \cite{yan2011approximation} who
gave a 3.16-approximation LP-rounding algorithm based on \cite{Shmoys97FL,Chudak0312e},
and recently claimed the ratio of 1.575 \cite{yan2012lp} built on
the work of \cite{Chudak0312e,jaroslaw2010optimal,jaroslaw2012lp,Guha03FTFL2.41}.
They aim to close the approximation gap between $FTRA_{\infty}$ and
$UFL$. On the other hand, due to the difficulties both inherited
from $FTFL$ and $FTRA_{\infty}$, it is still unknown what the approximation
gap between $FTRA$ and $FTFL$ is.

In this paper, we strive to close this gap. However, there are several
difficulties. First, despite the similar combinatorial structures
of $FTRA_{\infty}$ and $FTRA$, the existing LP-rounding algorithms
\cite{yan2011approximation,yan2012lp} for $FTRA_{\infty}$ can not
be adopted for $FTRA$. This is because these algorithms produce infeasible
solutions that violate the constraint $R_{i}$ in $FTRA$. In particular,
the recent work of \cite{yan2012lp} requires liberally splitting
facilities and randomly opening them. This can not be done for both
$FTRA$ and $FTFL$ as the splitting may cause more than $R_{i}$
facilities to open, which is not a problem for $FTRA_{\infty}$. Second,
in $FTFL$, $\max_{j\in\mathcal{C}}r_{j}\leq n_{f}$, while $r_{j}$
can be much larger than $n_{f}$ in both $FTRA_{\infty}$ and $FTRA$.
Therefore, the naive reduction idea of splitting the sites of an $FTRA$
instance and then restrict each site to have at most one facility
will create an equivalent $FTFL$ instance with a possibly exponential
size. Third, significantly more insights and heuristics are needed
in addition to the previous work for solving $FTRA$ (both the general
and the uniform cases) in polynomial time.

\textbf{Our Contribution.} For the general $FTRA$, we first develop
a \textit{unified LP-rounding algorithm }through modifying and extending
the 4-approximation LP-rounding algorithm \cite{Swamy08FTFL2.076}
for $FTFL$. The algorithm can directly solve $FTRA$, $FTRA_{\infty}$
and $FTFL$ with the same approximation ratio of 4. This is achieved
by: 1) constructing some useful \textit{properties} of the unified
algorithm which enable us to directly round the optimal fractional
solutions with values that might exceed one while ensuring the feasibility
of the rounded solutions and the algorithm correctness; 2) exploiting
the primal and dual complementary slackness conditions of the $FTRA$
problem's LP formulation. Then we show $FTRA$ reduces to $FTFL$
using an \textit{instance shrinking technique} inspired from the splitting
idea of \cite{yan2011newresults} for $FTRA_{\infty}$. It implies
that these two problems may share the same approximability in weakly
polynomial time. Hence, from the $FTFL$ result of \cite{JaroslawFTFL1.725},
we obtain the ratio of 1.7245\textcolor{black}{. For the non-metric
$FTRA$, we get }the first approximation factor of \textcolor{black}{$O\left(\log^{2}n\right)$
deduced from the work of \cite{Jain00FTFL,Lin92filting}. Note that,
although our first rounding algorithm attains a worse approximation
ratio, it could be more useful than the second to be adapted for other
variants of the resource allocation problems.}

For the uniform $FTRA$, better results are obtained. We first present
a naive \textit{primal-dual algorithm} that runs in pseudo-polynomial
time. To analyze the algorithm, we adopt a \textit{constraint-based
analysis} to derive the ratio of 1.61. Compared to dual fitting \cite{Jain03dualfitting}
and inverse dual fitting \cite{shihongftfa}, our analysis approach
is simpler and more convenient for handling more complicated dual
constructions. Later, with a carefully designed \textit{acceleration
heuristic} applied to the primal-dual algorithm, we obtain the first
strongly polynomial time algorithm for $FTRA$ that has the same ratio
of 1.61 but runtime \textcolor{black}{$O\left(n^{4}\right)$. }Moreover,
by applying another similar heuristic to the greedy augmentation technique
\cite{Guha03FTFL2.41}, \textcolor{black}{the 3.16-approximation rounding
result of \cite{yan2011approximation} for the general $FTRA_{\infty}$
is improved to 2.408, and the previous 1.61 ratio for the uniform
}$FTRA$ \textcolor{black}{reduces to 1.52.}

Lastly, we consider an important variant of $FTRA$ -- the Constrained
Fault-Tolerant $k$-Resource Allocation ($k$-$FTRA$) problem which
adds an extra global constraint that at most $k$ facilities across
all sites can be opened as resources. For the uniform $k$-$FTRA$,
based on the work of \cite{jain01approximation,Jain03dualfitting,Swamy08FTFL2.076},
we give the first constant-factor approximation algorithm for this
problem with a factor of 4. In particular, the algorithm relies on
a polynomial time \textit{greedy pairing} \textit{procedure} we develop
for efficiently splitting sites into paired and unpaired facilities.

The results shown directly hold for $FTRA_{\infty}$ and $ $$k$-$FTRA_{\infty}$,
and the techniques developed will be useful for other variants of
the resource allocation problems. For ease of analysis and implementation,
the algorithms presented mostly follow the pseudocode style. Furthermore,
we distinguish among pseudo-, weakly and strongly polynomial time
algorithms w.r.t. the problem size $n$.

\section{LP Basics and Properties}

The $FTRA$ problem has the following ILP formulation, in which \textcolor{black}{solution
variable $y_{i}$ denotes the number of facilities to open at site
$i$, and $x_{ij}$ the number of connections between client $j$
and site $i$. From the ILP, we can verify that the problem becomes
the special cases $FTFL$ if all $R_{i}$'s are uniform and equal
to $1$, and $FTRA_{\infty}$ if the third resource constraint is
removed.} 

\textit{\small 
\begin{equation}
\begin{array}{llc}
\mathrm{minimize} & \sum_{i\in\mathcal{F}}f_{i}y_{i}+\sum_{i\in\mathcal{F}}\sum_{j\in\mathcal{C}}c_{ij}x_{ij}\\
\mathrm{subject\, to} & \forall j\in\mathcal{C}:\,\sum_{i\in\mathcal{F}}x_{ij}\ge r_{j}\\
 & \forall i\in\mathcal{F},j\in\mathcal{C}:\, y_{i}-x_{ij}\geq0\\
 & \forall i\in\mathcal{F}:\, y_{i}\leq R_{i}\\
 & \forall i\in\mathcal{F},j\in\mathcal{C}:\, x_{ij},\, y_{i}\in\mathbb{Z}^{+}
\end{array}\label{eq:ftra-ip}
\end{equation}
}{\small \par}

The problem's LP-relaxation (primal LP) and dual LP are the following:

\textit{\small 
\begin{equation}
\begin{array}{llc}
\mathrm{minimize} & \sum_{i\in\mathcal{F}}f_{i}y_{i}+\sum_{i\in\mathcal{F}}\sum_{j\in\mathcal{C}}c_{ij}x_{ij}\\
\mathrm{subject\, to} & \forall j\in\mathcal{C}:\,\sum_{i\in\mathcal{F}}x_{ij}\ge r_{j}\\
 & \forall i\in\mathcal{F},j\in\mathcal{C}:\, y_{i}-x_{ij}\geq0\\
 & \forall i\in\mathcal{F}:\, y_{i}\leq R_{i}\\
 & \forall i\in\mathcal{F},j\in\mathcal{C}:\, x_{ij},\, y_{i}\geq0
\end{array}\label{eq:ftra-lp}
\end{equation}
}{\small \par}

{\small 
\begin{equation}
\begin{array}{llc}
\textrm{maximize} & \sum_{j\in\mathcal{C}}r_{j}\alpha_{j}-\sum_{i\in\mathcal{F}}R_{i}z_{i}\\
\mathrm{subject\, to} & \forall i\in\mathcal{F}:\,\sum_{j\in\mathcal{C}}\beta_{ij}\leq f_{i}+z_{i}\\
 & \forall i\in\mathcal{F},j\in\mathcal{C}:\,\alpha_{j}-\beta_{ij}\leq c_{ij}\\
 & \forall i\in\mathcal{F},j\in\mathcal{C}:\,\alpha_{j},\,\beta_{ij},\, z_{i}\geq0
\end{array}\label{eq:ftra-dual}
\end{equation}
}{\small \par}

Now we let $\left(\boldsymbol{x^{*}},\,\boldsymbol{y^{*}}\right)$
and $\left(\boldsymbol{\alpha^{*}},\,\boldsymbol{\beta^{*}},\,\boldsymbol{z^{*}}\right)$
be the optimal fractional primal and dual solutions of the LPs, and
$cost\left(\boldsymbol{x},\,\boldsymbol{y}\right)$ and $cost\left(\boldsymbol{\alpha},\,\boldsymbol{\beta},\,\boldsymbol{z}\right)$
be the cost functions (objective value functions) of any primal and
dual solutions respectively. By the strong duality theorem, $cost\left(\boldsymbol{x^{*}},\,\boldsymbol{y^{*}}\right)=cost\left(\boldsymbol{\alpha^{*}},\,\boldsymbol{\beta^{*}},\,\boldsymbol{z^{*}}\right)$.
Moreover, the primal complementary slackness conditions (CSCs) are:

(C1) If $x_{ij}^{*}>0\,$ then $\alpha_{j}^{*}=\beta_{ij}^{*}+c_{ij}$.

(C2) If $y_{i}^{*}>0\,$ then $\sum_{j\in\mathcal{C}}\beta_{ij}^{*}=f_{i}+z_{i}^{*}$.

Dual CSCs are:

(C3) If $\alpha_{j}^{*}>0\,$ then $\sum_{i\in\mathcal{F}}x_{ij}^{*}=r_{j}$.

(C4) If $\beta_{ij}^{*}>0\,$ then $x_{ij}^{*}=y_{i}^{*}$.

(C5) If $z_{i}^{*}>0\,$ then $y_{i}^{*}=R_{i}$.

W.l.o.g., $\left(\boldsymbol{x^{*}},\,\boldsymbol{y^{*}}\right)$
and $\left(\boldsymbol{\alpha^{*}},\,\boldsymbol{\beta^{*}},\,\boldsymbol{z^{*}}\right)$
have the following properties:

(P1) $\forall j\in\mathcal{C}:\,\alpha_{j}^{*}>0\,$ and $\sum_{i\in\mathcal{F}}x_{ij}^{*}=r_{j}$.

(P2) $\left(\boldsymbol{x^{*}},\,\boldsymbol{y^{*}}\right)$ is 'almost'
complete, i.e. $\forall j\in\mathcal{C}:$ if $x_{ij}^{*}>0$ then
$x_{ij}^{*}=y_{i}^{*}$ (the complete condition) or there is at most
one $i$ s.t. $0<x_{ij}^{*}<y_{i}^{*}$ where $i$ is the farthest
site connecting $j$. (cf. \cite{Chudak0312e,Swamy08FTFL2.076} for
more details)

\section{A Unified LP-Rounding Algorithm}

The algorithm ULPR (Algorithm 1) starts by solving the primal and
dual LPs to get the optimal solutions $\left(\boldsymbol{x^{*}},\,\boldsymbol{y^{*}}\right)$
and $\left(\boldsymbol{\alpha^{*}},\,\boldsymbol{\beta^{*}},\,\boldsymbol{z^{*}}\right)$
to work with. In order to utilize the dual LP for analyzing the approximation
ratio of the output solution $\left(\boldsymbol{x},\,\boldsymbol{y}\right)$,
we need to first deal with how to bound the $-\sum_{i\in\mathcal{F}}R_{i}z_{i}$
term in the dual objective function, introduced by imposing the new
resource constraint $\forall i\in\mathcal{F}:\, y_{i}\leq R_{i}$
in the primal LP. To resolve this, we exploit the dual CSC (C5). This
condition guides us to come up with Stage 1 of the algorithm ULPR
which fully opens all (facilities of) sites with $y_{i}^{*}=R_{i}$
and put these sites into the set $\mathcal{P}$ for pruning in the
future. Moreover, for successfully deriving the bound stated in Lemma
\ref{lem:1}, in the algorithm the client connections $x_{ij}^{*}$
with the opened sites in $\mathcal{P}$ $ $are rounded up to $\left\lceil x_{ij}^{*}\right\rceil $;
in the analysis the other primal and dual CSCs are also exploited.
At the end of Stage 1, for each $j$, we calculate its established
connection $\hat{r_{j}}$, residual connection requirement $\bar{r_{j}}$
and record its connected sites not in $\mathcal{P}$ as set $\mathcal{F}_{j}$
for the use of next stage.

\begin{algorithm}[H]
\caption{ULPR: Unified LP-Rounding Algorithm}

\textbf{\textcolor{black}{Input}}\textcolor{black}{: }$\mathcal{F},\,\mathcal{C},\,\boldsymbol{f},\,\boldsymbol{c},\,\boldsymbol{r},\,\boldsymbol{R}$.\textbf{\textcolor{black}{{}
Output: }}$\left(\boldsymbol{x},\,\boldsymbol{y}\right)$

\textbf{\textcolor{black}{Initialization}}\textcolor{black}{: }Solve
LPs \eqref{eq:ftra-lp} and \eqref{eq:ftra-dual} to obtain the optimal
fractional solutions $\left(\boldsymbol{x^{*}},\,\boldsymbol{y^{*}}\right)$
and $\left(\boldsymbol{\alpha^{*}},\,\boldsymbol{\beta^{*}},\,\boldsymbol{z^{*}}\right)$.
$\boldsymbol{x}\leftarrow\boldsymbol{0},\,\boldsymbol{y}\leftarrow\boldsymbol{0},\,\mathcal{P}\leftarrow\emptyset$

\textbf{\textcolor{black}{Stage 1}}\textcolor{black}{: Pruning and
Rounding}

\textbf{for} $i\in\mathcal{F}$

\qquad{}\textbf{if} $y_{i}^{*}=R_{i}$ \textbf{do}

\qquad{}\qquad{}$y_{i}\leftarrow R_{i}$

\qquad{}\qquad{}$\mathcal{P}\leftarrow\mathcal{P}\cup\left\{ i\right\} $

\qquad{}\qquad{}\textbf{for} $j\in\mathcal{C}$ 

\qquad{}\qquad{}\qquad{}\textbf{if} $x{}_{ij}^{*}>0$ \textbf{do}

\qquad{}\qquad{}\qquad{}\qquad{}$x_{ij}\leftarrow\left\lceil x_{ij}^{*}\right\rceil $

\textbf{set} $\forall j\in\mathcal{C}:\,\hat{r_{j}}\leftarrow\sum_{i\in\mathcal{P}}x_{ij},\,\bar{r_{j}}\leftarrow r_{j}-\hat{r_{j}},\,\mathcal{F}_{j}\leftarrow\left\{ i\in\mathcal{F}\backslash\mathcal{P}\,|\, x_{ij}^{*}>0\right\} $

\medskip{}

\textbf{\textcolor{black}{Stage 2}}\textcolor{black}{: Clustered Rounding}

\textbf{set} $\bar{\mathcal{C}}\leftarrow\left\{ j\in\mathcal{C}\,|\,\bar{r_{j}}\geq1\right\} $

\textbf{while }$\bar{\mathcal{C}}\neq\emptyset$

\qquad{}\textbf{//2.1}: Construct a cluster $\mathcal{S}$ centered
at $j_{o}$

\qquad{}$j_{o}\leftarrow\arg\min_{j}\left\{ \alpha_{j}^{*}:\, j\in\bar{\mathcal{C}}\right\} $,
order $\mathcal{F}_{j_{o}}$ by non-decreasing site facility costs

\qquad{}choose $\mathcal{S}\subseteq\mathcal{F}_{j_{o}}$ starting
from the cheapest site in $\mathcal{F}_{j_{o}}$ s.t. just $\sum_{i\in\mathcal{S}}y_{i}^{*}\geq\bar{r_{j_{o}}}$

\qquad{}\textbf{if} $\sum_{i\in\mathcal{S}}y_{i}^{*}>\bar{r_{j_{o}}}$
\textbf{do}

\qquad{}\qquad{}split the last most expensive site $i_{l}\in\mathcal{S}$
into $i_{1}$ and $i_{2}$: $y_{i_{1}}^{*}=\bar{r_{j_{o}}}-\sum_{i\in\mathcal{S}\backslash i_{l}}y_{i}^{*}$, 

\qquad{}\qquad{}$y_{i_{2}}^{*}=y_{i_{l}}^{*}-y_{i_{1}}^{*}$;\textbf{
forall $j$}:\textbf{ }set $x_{i_{1}j}^{*},\, x_{i_{2}j}^{*}$ s.t.
$x_{i_{1}j}^{*}+x_{i_{2}j}^{*}=x_{i_{l}j}^{*}$, $x_{i_{1}j}^{*}\leq y_{i_{1}}^{*}$ 

\qquad{}\qquad{}$x_{i_{2}j}^{*}\leq y_{i_{2}}^{*}$ and update $\mathcal{F}_{j}$;\textbf{
}$\mathcal{S}\leftarrow\mathcal{S}\backslash\left\{ i_{l}\right\} \cup\left\{ i_{1}\right\} $
(now $\sum_{i\in\mathcal{S}}y_{i}^{*}=\bar{r_{j_{o}}}$)\medskip{}

\qquad{}//\textbf{2.2}: Rounding around $j_{o}$ and $\mathcal{S}$

\qquad{}//\textbf{2.2.1}: Finish rounding $\boldsymbol{y}$

\qquad{}\textbf{for $i\in\mathcal{S}$ }//from the cheapest site

\qquad{}\qquad{}$y_{i}\leftarrow\left\lceil y_{i}^{*}\right\rceil $ 

\qquad{}\qquad{}$\bar{\mathcal{S}}\leftarrow\bar{\mathcal{S}}\cup\left\{ i\right\} $
//maintain a set of already rounded sites

\qquad{}\qquad{}\textbf{if $\sum_{i'\in\bar{\mathcal{S}}}y_{i'}\geq\bar{r_{j_{o}}}$}

\qquad{}\qquad{}\qquad{}$y_{i}\leftarrow\bar{r_{j_{o}}}-\sum_{i'\in\bar{\mathcal{S}}\backslash i}y_{i'}$
(resetting $y_{i}$ to make $\sum_{i'\in\bar{\mathcal{S}}}y_{i}=\bar{r_{j_{o}}}$)

\qquad{}\qquad{}\qquad{}\textbf{break}\medskip{}

\qquad{}//\textbf{2.2.2}: Finish rounding $\boldsymbol{x}$

\qquad{}\textbf{for} $j\in\bar{\mathcal{C}}$ //including $j_{o}$

\qquad{}\qquad{}if $\mathcal{F}_{j}\cap\mathcal{S}\neq\emptyset$

\qquad{}\qquad{}\qquad{}\textbf{for $i\in\bar{\mathcal{S}}$ }//order
does not matter, could connect to the closest

\qquad{}\qquad{}\qquad{}\qquad{}$x_{ij}\leftarrow\min\left(\bar{r_{j}},\, y_{i}\right)$

\qquad{}\qquad{}\qquad{}\qquad{}$\bar{r_{j}}\leftarrow\bar{r_{j}}-x_{ij}$

\qquad{}\qquad{}$\mathcal{F}_{j}=\mathcal{F}_{j}\backslash\mathcal{S}$

\qquad{}\textbf{update $\bar{\mathcal{C}}$}
\end{algorithm}

While most LP-rounding algorithms round optimal solutions with values
in $\left[0,\,1\right]$, for $FTRA$, our approach directly rounds
the solutions with values in $\left[0,\, R_{i}\right]$ and later
we shall analyze its correctness via establishing some useful properties.
Like the major LP-rounding algorithms \cite{Sviridenko02improved1.58,jaroslaw2010optimal,jaroslaw2012lp,Li2011}
for $UFL$, Stage 2 of our algorithm also inherits the classical iterative
clustering idea \cite{Shmoys97FL,Chudak0312e}. The clustering and
rounding here terminate when all $\bar{r_{j}}$'s are satisfied, i.e.
the set of not-fully-connected clients $\bar{\mathcal{C}}=\emptyset$
in the algorithm. Stage 2 consists of two substages 2.1 and 2.2, dealing
with cluster construction and cluster guided rounding respectively.
Stage 2.1 essentially adopts the facility cloning idea \cite{Swamy08FTFL2.076}
for the deterministic rounding of $FTFL$. Nevertheless, here we are
splitting sites. In each iteration, it first picks the cluster center
$j_{o}$ with the smallest optimal dual value, and then builds a cluster
$\mathcal{S}$ around it which contains a subset of ordered sites
in $\mathcal{F}_{j_{o}}$, starting from the cheapest site until $\sum_{i\in\mathcal{S}}y_{i}^{*}\geq\bar{r_{j_{o}}}$.
In order to maintain the invariant $\forall j\in\bar{\mathcal{C}}:\,\sum_{i\in\mathcal{F}_{j}}y_{i}^{*}\geq\bar{r_{j}}$
in every iteration, the stage then splits the last site $i_{l}\in\mathcal{S}$
into $i_{1}$ and $i_{2}$, updates the client connections w.r.t.
$i_{1}$ and $i_{2}$, and in $\mathcal{S}$ includes $i_{1}$ while
excluding $i_{l}$ to keep $\sum_{i\in\mathcal{S}}y_{i}^{*}=\bar{r_{j_{o}}}$.
Stage 2.2 does the final rounding steps around $\mathcal{S}$ in addition
to Stage 1 to produce a feasible integral solution $\left(\boldsymbol{x},\,\boldsymbol{y}\right)$.
This stage modifies and generalizes the rounding steps for $FTFL$.
Its substage 2.2.1 rounds up the sites ($y_{i}^{*}\rightarrow\left\lceil y_{i}^{*}\right\rceil $)
from the cheapest site in $\mathcal{S}$ until $\bar{\mathcal{S}}$
(the set of sites rounded so far) just satisfies \textbf{$\sum_{i'\in\bar{\mathcal{S}}}y_{i'}\geq\bar{r_{j_{o}}}$}
(now these $y_{i'}$'s are already integral). To make sure $\sum_{i'\in\bar{\mathcal{S}}}y_{i}=\bar{r_{j_{o}}}$
for bounding the site facility opening cost (cf. Lemma \ref{lem: boc}),
the integral facility opening $y_{i}$ of the last site $i$ in $\bar{\mathcal{S}}$
is reset to $\bar{r_{j_{o}}}-\sum_{i'\in\bar{\mathcal{S}}\backslash i}y_{i'}$,
which is also integral. After the facilities at the sites in $\bar{\mathcal{S}}$
are opened according to the $y_{i}$'s constructed in stage 2.2.1,
stage 2.2.2 then connects every client $j$ in $\bar{\mathcal{C}}$
which has connections to the sites in $\mathcal{S}$ (according to
the $x_{ij}^{*}$'s) to $\min\left(\bar{r_{j}},\,\bar{r_{j_{o}}}\right)$
of these open facilities. It does this by iterating through all sites
in $\bar{\mathcal{S}}$, setting $x_{ij}$'s and updating $\bar{r_{j}}$'s
as described in the algorithm. At the end, for the run of next iteration,
the sites in the cluster $\mathcal{S}$ are excluded from $\mathcal{F}_{j}$,
implying all clusters chosen in the iterations are disjoint; and $\bar{\mathcal{C}}$
is updated (at least $j_{o}$ is removed from the set). 

In the analysis, we first demonstrate the overall correctness of the
algorithm ensured by the following properties. Note that some of the
proofs in this section frequently refer to the content of Section
2.

\medskip{}

(P3) After Stage 1, $ $$\forall i\in\mathcal{P},\, j\in\mathcal{C}:$
$x_{ij}\leq R_{i}$ and $\bar{r_{j}}=r_{j}-\hat{r_{j}}\geq0$.
\begin{proof}
The first part of the property is obvious since $x_{ij}=0$ or $x_{ij}=\left\lceil x_{ij}^{*}\right\rceil \leq\left\lceil y_{i}^{*}\right\rceil \leq R_{i}$.
For the second part, $\forall j\in\mathcal{C}:$ if all $x_{ij}^{*}$
are integers, we are done. Now we only need to consider $j$'s fractional
$x_{ij}^{*}$ connecting with $ $$\mathcal{P}$. By the previous
property (P2), there is at most one fractional $x_{ij}^{*}$ with
$\mathcal{P}$ because all $y_{i}^{*}$'s in $\mathcal{P}$ are integers.
Therefore, in Stage 1, at most one fractional $x_{ij}^{*}$ is rounded
up which will not make $\hat{r_{j}}$ exceed $r_{j}$.
\end{proof}

(P4) Stage 2.2.1 rounds $y_{i_{1}}^{*}$ (the optimal fractional opening
of the last site $i_{1}$ in $\mathcal{S}$ which is included in Stage
2.1) to at most $\left\lfloor y_{i_{1}}^{*}\right\rfloor $. 
\begin{proof}
If $y_{i_{1}}^{*}$ is integral, the property clearly holds. Otherwise
if $y_{i_{1}}^{*}$ is fractional, Case 1): if \textbf{$\sum_{i'\in\bar{\mathcal{S}}}y_{i'}=\bar{r_{j_{o}}}$
}before resetting the last site $i$ in Stage 2.2.1, this $i$ definitely
appears before $i_{1}$ in $\mathcal{S}$ because otherwise $\sum_{i'\in\bar{\mathcal{S}}}\left\lceil y_{i'}^{*}\right\rceil $
will exceed $\bar{r_{j_{o}}}$, therefore $i_{1}$ is left unrounded;
Case 2): If \textbf{$\sum_{i'\in\bar{\mathcal{S}}}y_{i'}>\bar{r_{j_{o}}}$},
the last site $i$ is possibly $i_{1}$ and if it is then $\bar{\mathcal{S}}=\mathcal{S}$,
and from the algorithm we have rounded $y_{i_{1}}=\bar{r_{j_{o}}}-\sum_{i'\in\mathcal{S}\backslash i_{1}}\left\lceil y_{i'}^{*}\right\rceil $
after resetting. If $y_{i_{1}}>\left\lfloor y_{i_{1}}^{*}\right\rfloor $
we get $\sum_{i'\in\mathcal{S}\backslash i_{1}}\left\lceil y_{i'}^{*}\right\rceil +\left\lfloor y_{i_{1}}^{*}\right\rfloor <\bar{r_{j_{o}}}$
which is not possible since $\sum_{i\in\mathcal{S}}y_{i}^{*}=\bar{r_{j_{o}}}=\sum_{i'\in\mathcal{S}\backslash i_{1}}y_{i'}^{*}+y_{i_{1}}^{*}=\left\lceil \sum_{i'\in\mathcal{S}\backslash i_{1}}y_{i'}^{*}\right\rceil +\left\lfloor y_{i_{1}}^{*}\right\rfloor \leq\sum_{i'\in\mathcal{S}\backslash i_{1}}\left\lceil y_{i'}^{*}\right\rceil +\left\lfloor y_{i_{1}}^{*}\right\rfloor $
(because $y_{i_{1}}^{*}$ is fractional). Hence, $ $$y_{i_{1}}^{*}$
is rounded to at most $\left\lfloor y_{i_{1}}^{*}\right\rfloor $.
\end{proof}

(P5) $\forall i\in\mathcal{F}:$ given $0<y_{i_{1}}^{*}+y_{i_{2}}^{*}=y_{i}^{*}\leq R_{i}$,
then we have $\left\lfloor y_{i_{1}}^{*}\right\rfloor +\left\lceil y_{i_{2}}^{*}\right\rceil \leq\left\lceil y_{i}^{*}\right\rceil \leq R_{i}$.
\begin{proof}
We first have $ $$\left\lfloor y_{i_{1}}^{*}\right\rfloor \leq y_{i_{1}}^{*}$
and $\left\lceil y_{i_{2}}^{*}\right\rceil <y_{i_{2}}^{*}+1$, so
$\left\lfloor y_{i_{1}}^{*}\right\rfloor +\left\lceil y_{i_{2}}^{*}\right\rceil <y_{i_{1}}^{*}+y_{i_{2}}^{*}+1=y_{i}^{*}+1$.
Now if $y_{i}^{*}$ is integral, because $\left\lfloor y_{i_{1}}^{*}\right\rfloor +\left\lceil y_{i_{2}}^{*}\right\rceil $
is also integral, $\left\lfloor y_{i_{1}}^{*}\right\rfloor +\left\lceil y_{i_{2}}^{*}\right\rceil \leq\left\lceil y_{i}^{*}\right\rceil $.
Otherwise if $y_{i}^{*}$ is fractional, $\left\lfloor y_{i_{1}}^{*}\right\rfloor +\left\lceil y_{i_{2}}^{*}\right\rceil \leq\left\lfloor y_{i}^{*}\right\rfloor +1=\left\lceil y_{i}^{*}\right\rceil $.
The property then follows from $\forall i\in\mathcal{F}:\,\left\lceil y_{i}^{*}\right\rceil \leq R_{i}$.
\end{proof}
\medskip{}

In summary, property (P3) shows the correctness of Stage 1 before
going into Stage 2, (P4) and (P5) together ensure the splitting in
Stage 2.1 and the rounding in Stage 2.2.1 produce feasible $y_{i}$'s
for $FTRA$. This is because for any split sites $i_{1}$ and $i_{2}$
from $i$, (P4) guarantees at $i_{1}$ at most $\left\lfloor y_{i_{1}}^{*}\right\rfloor $
facilities are open, and (P5) makes sure that even $\left\lceil y_{i_{2}}^{*}\right\rceil $
facilities are opened at $i_{2}$ in the subsequent iterations of
the algorithm, no more than $R_{i}$ facilities in total actually
get opened at $i$. Note that, (P5) also covers the situation that
a site is repeatedly (recursively) split. Furthermore, in each iteration,
Stage 2 at least fully connects the client $j_{o}$ and considers
all sites in the cluster $\mathcal{S}$ centered at $j_{o}$. More
importantly, the invariant $\forall j\in\bar{\mathcal{C}}:\,\sum_{i\in\mathcal{F}_{j}}y_{i}^{*}\geq\bar{r_{j}}$
is maintained for choosing the feasible cluster $\mathcal{S}$ in
Stage 2.1. This is true in the first iteration. In the subsequent
iterations, the invariant still preserves because for any $j$ with
$\mathcal{F}_{j}\cap\mathcal{S}\neq\emptyset$ that is not fully connected
in the current iteration, in the next iteration, $\sum_{i\in\mathcal{F}_{j}}y_{i}^{*}$
is decreased by at most $\bar{r_{j_{o}}}$ (because Stage 2.1 splits
sites to maintain $\sum_{i\in\mathcal{S}}y_{i}^{*}=\bar{r_{j_{o}}}$
and $\mathcal{S}$ is excluded from $\mathcal{F}_{j}$ in Stage 2.2.2)
and $\bar{r_{j}}$ is decreased by exactly $\bar{r_{j_{o}}}$ (from
Stage 2.2.2). Therefore, the overall algorithm is correct.

Furthermore, the time complexity of the rounding stages of Algorithm
1 is $O\left(n^{3}\right)$ since each iteration of Stage 2 at least
fully connects one of $n_{c}$ clients which takes time $O\left(n^{2}\right)$.
In the following, we separately bound the partial solution costs incurred
in the stages involving rounding and then combine these costs for
achieving the approximation ratio.
\begin{lemma}
After pruning and rounding, the partial total cost from Stage 1 is
$\sum_{j\in\mathcal{C}}\hat{r_{j}}\alpha_{j}^{*}-\sum_{i\in\mathcal{F}}R_{i}z_{i}^{*}$.
\label{lem:1}\end{lemma}
\begin{proof}
$\forall i\in\mathcal{P}:\,$

\begin{eqnarray*}
 & \sum_{j\in\mathcal{C}}\left\lceil x_{ij}^{*}\right\rceil \alpha_{j}^{*} & =\sum_{j\in\mathcal{C}}\left\lceil x_{ij}^{*}\right\rceil c_{ij}+\sum_{j\in\mathcal{C}}\left\lceil x_{ij}^{*}\right\rceil \beta_{ij}^{*}\\
 &  & =\sum_{j\in\mathcal{C}}\left\lceil x_{ij}^{*}\right\rceil c_{ij}+\sum_{j:\, x_{ij}^{*}=y_{i}^{*}=R_{i}}\left\lceil x_{ij}^{*}\right\rceil \beta_{ij}^{*}+\sum_{j:\, x_{ij}^{*}<y_{i}^{*}=R_{i}}\left\lceil x_{ij}^{*}\right\rceil \beta_{ij}^{*}\\
 &  & =\sum_{j\in\mathcal{C}}\left\lceil x_{ij}^{*}\right\rceil c_{ij}+\sum_{j:\, x_{ij}^{*}=y_{i}^{*}=R_{i}}\left\lceil x_{ij}^{*}\right\rceil \beta_{ij}^{*}\\
 &  & =\sum_{j\in\mathcal{C}}\left\lceil x_{ij}^{*}\right\rceil c_{ij}+R_{i}\left(\sum_{j:\, x_{ij}^{*}=y_{i}^{*}=R_{i}}\beta_{ij}^{*}+\sum_{j:\, x_{ij}^{*}<y_{i}^{*}=R_{i}}\beta_{ij}^{*}\right)\\
 &  & =\sum_{j\in\mathcal{C}}\left\lceil x_{ij}^{*}\right\rceil c_{ij}+R_{i}\sum_{j\in\mathcal{C}}\beta_{ij}^{*}\\
 &  & =\sum_{j\in\mathcal{C}}\left\lceil x_{ij}^{*}\right\rceil c_{ij}+R_{i}f_{i}+R_{i}z_{i}^{*}.
\end{eqnarray*}

The first equality is due to the condition (C1), the third, fourth
and fifth is because by (C4) we have $\forall i\in\mathcal{P},\, j\in\mathcal{C}:$
if $\beta_{ij}^{*}>0\,$ then $x_{ij}^{*}=y_{i}^{*}=R_{i}$, so $x_{ij}^{*}<y_{i}^{*}=R_{i}$
implies $\beta_{ij}^{*}=0$ (by the contraposition in logic) and also
$\sum_{j:\, x_{ij}^{*}<y_{i}^{*}=R_{i}}\beta_{ij}^{*}=0$. The last
equality is obtained from (C2), and the fact that $\forall i\in\mathcal{P}:\, y_{i}^{*}=R_{i}>0$.

Summing both sides over all $i\in\mathcal{P}$, we can then bound
the cost of Stage 1:

$ $
\begin{eqnarray*}
 & \sum_{i\in\mathcal{P}}\sum_{j\in\mathcal{C}}\left\lceil x_{ij}^{*}\right\rceil c_{ij}+\sum_{i\in\mathcal{P}}R_{i}f_{i} & =\sum_{i\in\mathcal{P}}\sum_{j\in\mathcal{C}}\left\lceil x_{ij}^{*}\right\rceil \alpha_{j}^{*}-\sum_{i\in\mathcal{P}}R_{i}z_{i}^{*}\\
 &  & =\sum_{j\in\mathcal{C}}\hat{r_{j}}\alpha_{j}^{*}-\sum_{i\in\mathcal{F}}R_{i}z_{i}^{*}.
\end{eqnarray*}

The second equality follows from Stage 1 that $\forall j\in\mathcal{C}:\,\hat{r_{j}}=\sum_{i\in\mathcal{P}}\left\lceil x_{ij}^{*}\right\rceil $,
and the condition (C5): if $z_{i}^{*}>0\,$ then $y_{i}^{*}=R_{i}$,
so $y_{i}^{*}<R_{i}$ implies $z_{i}^{*}=0\,$.\end{proof}
\begin{lemma}
After rounding $\boldsymbol{y}$, the partial site facility opening
cost from Stage 2.2.1 is at most $\sum_{i\in\mathcal{F}\backslash\mathcal{P}}f_{i}y_{i}^{*}$.
\label{lem: boc}\end{lemma}
\begin{proof}
Facilities at sites $i\in{\cal S}\subseteq\mathcal{F}\backslash\mathcal{P}$
are opened in $f_{i}$'s non-decreasing order in Stage 2.2.1: In any
iteration of the algorithm with picked cluster $\mathcal{S}$, before
rounding we have $\sum_{i\in\mathcal{S}}y_{i}^{*}=\bar{r_{j_{o}}}$;
after rounding set $\bar{\mathcal{S}}$ is formed starting from the
cheapest site in $\mathcal{S}$ s.t. $\sum_{i'\in\bar{\mathcal{S}}}y_{i}=\bar{r_{j_{o}}}$.
This makes the opening cost of all sites in cluster $\bar{{\cal S}}$
at most $\sum_{i\in\mathcal{S}}f_{i}y_{i}^{*}$. The lemma then follows
from the fact that all chosen clusters are disjoint in the algorithm.\end{proof}
\begin{lemma}
After rounding $\boldsymbol{x}$, the partial connection cost from
Stage 2.2.2 is at most $3\sum_{j\in\mathcal{C}}\bar{r_{j}}\alpha_{j}^{*}$.
\label{lem: conn} \end{lemma}
\begin{proof}
Let site $i$ lie in the cluster $\mathcal{S}$ centered at $j_{o}$.
If $j$ is already connected to $i$ ($x_{ij}^{*}>0$), then $c_{ij}\leq\alpha_{j}^{*}$
from the condition (C1). Otherwise, if $j$ connects to $i$ after
rounding, from the algorithm, it implies $\alpha_{j_{o}}^{*}\leq\alpha_{j}^{*}$
(because $j$ with the smallest $\alpha_{j}^{*}$ is always chosen
as $j_{o}$) and $\mathcal{F}_{j}\cap\mathcal{S}\neq\emptyset$. Fig.
2 then displays the case $\mathcal{F}_{j}\cap\mathcal{S}\neq\emptyset$
where initially $j$ connects to $i'$ and it is connected to $i$
after rounding. By the triangle inequality, we have $c_{ij}\leq c_{i'j}+c_{ij_{o}}+c_{i'j_{o}}$.
Also, it is true that $x_{i'j}^{*},\, x_{ij_{o}}^{*},\, x_{i'j_{o}}^{*}>0$,
so from (C1) we have $c_{i'j}\leq\alpha_{j}^{*}$ and $ $$c_{ij_{o}},\, c_{i'j_{o}}\leq\alpha_{j_{o}}^{*}$.
Hence, $c_{ij}\leq\alpha_{j}^{*}+2\alpha_{j_{o}}^{*}\leq3\alpha_{j}^{*}$.
Since each $j$ makes $\bar{r_{j}}$ connections, the total partial
connection cost bound is $3\sum_{j\in\mathcal{C}}\bar{r_{j}}\alpha_{j}^{*}$
and the lemma follows. Note that Fig. 2 does not show multiplicity
of the connection between any client and site in an $FTRA$ solution.
It is merely for simplicity and will not affect the correctness of
the proof.
\begin{figure}
\begin{centering}
\includegraphics{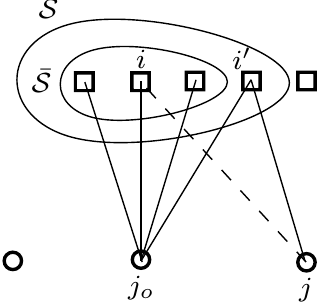}
\par\end{centering}

\caption{Illustration of bounding the connection costs}
\end{figure}
\end{proof}
\begin{theorem}
Algorithm ULPR is 4-approximation for $FTRA$.\end{theorem}
\begin{proof}
Adding up the partial cost bounds stated in the previous lemmas, the
total cost $cost\left(\boldsymbol{x},\,\boldsymbol{y}\right)$ is
therefore at most $\sum_{j\in\mathcal{C}}\hat{r_{j}}\alpha_{j}^{*}-\sum_{i\in\mathcal{F}}R_{i}z_{i}^{*}+\sum_{i\in\mathcal{F}\backslash\mathcal{P}}f_{i}y_{i}^{*}+3\sum_{j\in\mathcal{C}}\bar{r_{j}}\alpha_{j}^{*}$.
Also, we have $cost\left(\boldsymbol{x^{*}},\,\boldsymbol{y^{*}}\right)=\sum_{i\in\mathcal{F}}f_{i}y_{i}^{*}+\sum_{i\in\mathcal{F}}\sum_{j\in\mathcal{C}}c_{ij}x_{ij}^{*}=cost\left(\boldsymbol{\alpha^{*}},\,\boldsymbol{\beta^{*}},\,\boldsymbol{z^{*}}\right)=\sum_{j\in\mathcal{C}}r_{j}\alpha_{j}^{*}-\sum_{i\in\mathcal{F}}R_{i}z_{i}^{*}=\sum_{j\in\mathcal{C}}\hat{r_{j}}\alpha_{j}^{*}+\sum_{j\in\mathcal{C}}\bar{r_{j}}\alpha_{j}^{*}-\sum_{i\in\mathcal{F}}R_{i}z_{i}^{*}$,
so $cost\left(\boldsymbol{x},\,\boldsymbol{y}\right)\leq\sum_{j\in\mathcal{C}}\hat{r_{j}}\alpha_{j}^{*}-\sum_{i\in\mathcal{F}}R_{i}z_{i}^{*}+\left(\sum_{j\in\mathcal{C}}\hat{r_{j}}\alpha_{j}^{*}+\sum_{j\in\mathcal{C}}\bar{r_{j}}\alpha_{j}^{*}-\sum_{i\in\mathcal{F}}R_{i}z_{i}^{*}\right)+3\sum_{j\in\mathcal{C}}\bar{r_{j}}\alpha_{j}^{*}=4\sum_{j\in\mathcal{C}}\bar{r_{j}}\alpha_{j}^{*}+2\sum_{j\in\mathcal{C}}\hat{r_{j}}\alpha_{j}^{*}-2\sum_{i\in\mathcal{F}}R_{i}z_{i}^{*}$$=4cost\left(\boldsymbol{\alpha^{*}},\,\boldsymbol{\beta^{*}},\,\boldsymbol{z^{*}}\right)-2\left(\sum_{j\in\mathcal{C}}\hat{r_{j}}\alpha_{j}^{*}-\sum_{i\in\mathcal{F}}R_{i}z_{i}^{*}\right)\leq4cost\left(\boldsymbol{x^{*}},\,\boldsymbol{y^{*}}\right)$.
The last inequality follows from the fact that $\left(\sum_{j\in\mathcal{C}}\hat{r_{j}}\alpha_{j}^{*}-\sum_{i\in\mathcal{F}}R_{i}z_{i}^{*}\right)$
is the cost of Stage 1 (cf. Lemma \ref{lem:1}) which is nonnegative.
\end{proof}

\section{Reduction to $FTFL$}

Recently, the authors in \cite{yan2011newresults} presented a splitting
idea that is able to reduce any $FTRA_{\infty}$ instance with arbitrarily
large $r_{j}$ to another small $FTRA_{\infty}$ instance with polynomially
bounded $r_{j}$. The direct consequence of this is that $FTRA_{\infty}$
is then reducible to $FTFL$, since we are able to naively split the
sites of the small $FTRA_{\infty}$ instance and the resulting instance
is equivalent to an $FTFL$ instance with a polynomial size. Because
$FTRA$ and $FTRA_{\infty}$ have similar combinatorial structures
where $FTRA_{\infty}$ is a special case, the question then becomes
whether the more general $FTRA$ reduces to $FTFL$ as well. In the
following, we give an affirmative answer to this with an instance
shrinking technique.

Compared to the reduction in \cite{yan2011newresults} for $FTRA_{\infty}$,
first, the instance shrinking technique is more general. This is because
the technique reduces any $FTRA$ instance with arbitrarily large
$R_{i}$ to another small $FTRA$ instance with polynomially bounded
$R_{i}^{s}$, which works for $FTRA_{\infty}$ as well since an $FTRA_{\infty}$
instance can be treated as an $FTRA$ instance with all $R_{i}$'s
set to be $\max_{j\in\mathcal{C}}r_{j}$. The small $FTRA$ instance
is then equivalent to an $FTFL$ instance with a polynomial size ($\sum_{i\in\mathcal{F}}R_{i}^{s}$),
implying $FTRA$ and $FTFL$ share the same approximability in weakly
polynomial time. Second, the reduction for $FTRA_{\infty}$ does not
have a mechanism for bounding $R_{i}^{s}$ polynomially in $FTRA$.
Therefore, it can not directly yield a reduction result for $FTRA$.
On the other hand, our technique initially includes the following
\textit{crucial} instance shrinking mechanism for bounding $R_{i}^{s}$.
\begin{claim}
$\left(\boldsymbol{x^{*}},\,\boldsymbol{y^{*}}\right)$ remains to
be the optimal solution even if $R_{i}$ is replaced with $\left\lceil y_{i}^{*}\right\rceil $
in LP \eqref{eq:ftra-lp}. \end{claim}
\begin{proof}
Denote the instance with parameter $R_{i}$ as $\mathcal{I}_{o}$,
and $\mathcal{I}$ after replacing $R_{i}$ with $\left\lceil y_{i}^{*}\right\rceil $.
On one hand, solving $\mathcal{I}$ will not yield any better optimal
solution $\left(\tilde{\boldsymbol{x^{*}}},\,\tilde{\boldsymbol{y^{*}}}\right)$
with $cost\left(\tilde{\boldsymbol{x^{*}}},\,\tilde{\boldsymbol{y^{*}}}\right)<cost\left(\boldsymbol{x^{*}},\,\boldsymbol{y^{*}}\right)$,
because this $\left(\tilde{\boldsymbol{x^{*}}},\,\tilde{\boldsymbol{y^{*}}}\right)$
is also feasible to $\mathcal{I}_{o}$, which contradicts the optimality
of $\left(\boldsymbol{x^{*}},\,\boldsymbol{y^{*}}\right)$ for $\mathcal{I}_{o}$.
On the other hand, $cost\left(\tilde{\boldsymbol{x^{*}}},\,\tilde{\boldsymbol{y^{*}}}\right)>cost\left(\boldsymbol{x^{*}},\,\boldsymbol{y^{*}}\right)$
is not possible since $\left(\boldsymbol{x^{*}},\,\boldsymbol{y^{*}}\right)$
is also a feasible solution to $\mathcal{I}$ as $y_{i}^{*}\leq\left\lceil y_{i}^{*}\right\rceil $,
which contradicts the optimality of $\left(\tilde{\boldsymbol{x^{*}}},\,\tilde{\boldsymbol{y^{*}}}\right)$
for $\mathcal{I}$. Hence, $\left(\boldsymbol{x^{*}},\,\boldsymbol{y^{*}}\right)$
stays optimal for $\mathcal{I}$.
\end{proof}
With this mechanism, instead we can consider the equivalent $FTRA$
instance $\mathcal{I}$ with $\forall i\in\mathcal{F}:\, R_{i}=\left\lceil y_{i}^{*}\right\rceil $
and the same optimal solution $\left(\boldsymbol{x^{*}},\,\boldsymbol{y^{*}}\right)$.
Then in the reduction, $\left(\boldsymbol{x^{*}},\,\boldsymbol{y^{*}}\right)$
is split into a large integral solution with $y_{i}^{l}=\max\left(0,\,\left\lfloor y_{i}^{*}\right\rfloor -1\right)$
and $x_{ij}^{l}=\min\left(\left\lfloor x_{ij}^{*}\right\rfloor ,\, y_{i}^{l}\right)$,
and a small fractional solution with $ $$y_{i}^{s}=y_{i}^{*}-y_{i}^{l}$
and $x_{ij}^{s}=x_{ij}^{*}-x_{ij}^{l}$, for all $i\in\mathcal{F},\, j\in\mathcal{C}$.
Let the tuple $\left\langle \mathcal{F},\,\mathcal{C},\,\boldsymbol{f},\,\boldsymbol{c},\,\boldsymbol{r},\,\boldsymbol{R}\right\rangle $
represent the instance $\mathcal{I}$, the reduction then proceeds
by splitting $\mathcal{I}$ into a large instance $ $$\mathcal{I}^{l}$:
$\left\langle \mathcal{F},\,\mathcal{C},\,\boldsymbol{f},\,\boldsymbol{c},\,\boldsymbol{r^{l}},\,\boldsymbol{R^{l}}\right\rangle $
and a small instance $ $$\mathcal{I}^{s}$: $\left\langle \mathcal{F},\,\mathcal{C},\,\boldsymbol{f},\,\boldsymbol{c},\,\boldsymbol{r^{s}},\,\boldsymbol{R^{s}}\right\rangle $
according to $\left(\boldsymbol{x^{l}},\,\boldsymbol{y^{l}}\right)$
and $\left(\boldsymbol{x^{s}},\,\boldsymbol{y^{s}}\right)$. In particular,
these two instances differ at two parameters $\boldsymbol{r}$ and
$\boldsymbol{R}$, where we let $r_{j}^{l}=\sum_{i\in\mathcal{F}}x_{ij}^{l}$,
$r_{j}^{s}=r_{j}-r_{j}^{l}$ and $R_{i}^{l}=y_{i}^{l}$, $R_{i}^{s}=\left\lceil y_{i}^{*}\right\rceil -y_{i}^{l}$.
Note that although the above splitting idea of the instance shrinking
technique is inspired from the reduction for $FTRA_{\infty}$, the
focus on splitting $R_{i}$ is essentially different from reducing
$r_{j}$. Also, here we can see that the construction of the $ $shrunken
instance $\mathcal{I}$ with $R_{i}=\left\lceil y_{i}^{*}\right\rceil $
is crucial for bounding $R_{i}^{s}$, since if the original $R_{i}$
is used, $R_{i}^{s}$ can not be bounded and the technique will not
work. In the following, the first lemma mostly results from the original
splitting idea where we provide a simpler proof for it. The second
is directly from our instance shrinking and splitting on $R_{i}$.
As shown later in the proof of Theorem \ref{reduction-theorem}, these
lemmas are necessary for the approximation preserving reduction from
$\mathcal{I}$ to $\mathcal{I}^{s}$.
\begin{lemma}
$\left(\boldsymbol{x^{l}},\,\boldsymbol{y^{l}}\right)$ is a feasible
integral solution to $\mathcal{I}^{l}$ and $\left(\boldsymbol{x^{s}},\,\boldsymbol{y^{s}}\right)$
is a feasible fractional solution to $\mathcal{I}^{s}$.\label{lem:splitting}\end{lemma}
\begin{proof}
According to the LP \eqref{eq:ftra-lp}, it is trivial to see the
feasibility of the integral solution $\left(\boldsymbol{x^{l}},\,\boldsymbol{y^{l}}\right)$.
For the fractional solution $\left(\boldsymbol{x^{s}},\,\boldsymbol{y^{s}}\right)$,
since $ $$r_{j}=\sum_{i\in\mathcal{F}}x_{ij}^{*}$, $r_{j}^{l}=\sum_{i\in\mathcal{F}}x_{ij}^{l}$,
$r_{j}^{s}=r_{j}-r_{j}^{l}$ and $x_{ij}^{s}=x_{ij}^{*}-x_{ij}^{l}$,
we have $r_{j}^{s}=\sum_{i\in\mathcal{F}}x_{ij}^{s}$ and the first
constraint of the LP holds. Further, it is easy to see $y_{i}^{s}\leq R_{i}^{s}$
and we are left to show the second constraint $\forall i\in\mathcal{F},\, j\in\mathcal{C}:\, y_{i}^{s}-x_{ij}^{s}\geq0$
holds, i.e. $y_{i}^{*}-y_{i}^{l}\geq x_{ij}^{*}-\min\left(\left\lfloor x_{ij}^{*}\right\rfloor ,\, y_{i}^{l}\right)$.
Consider two cases: 1) $y_{i}^{l}\leq\left\lfloor x_{ij}^{*}\right\rfloor $,
then the inequality obviously follows from $y_{i}^{*}\geq x_{ij}^{*}$;
2) $y_{i}^{l}>\left\lfloor x_{ij}^{*}\right\rfloor $, the inequality
$rhs=x_{ij}^{*}-\left\lfloor x_{ij}^{*}\right\rfloor $, and $lhs=y_{i}^{*}-\max\left(0,\,\left\lfloor y_{i}^{*}\right\rfloor -1\right)$
after substituting $y_{i}^{l}$. Now again consider two sub cases:
2.1) $\left\lfloor y_{i}^{*}\right\rfloor \geq1$, then $lhs\geq1$
while $rhs\leq1$, so $lhs\geq rhs$ and the inequality follows; 2.2)
$\left\lfloor y_{i}^{*}\right\rfloor <1$, then $lhs=y_{i}^{*}$,
and since $1>y_{i}^{*}\geq x_{ij}^{*}$, $\left\lfloor x_{ij}^{*}\right\rfloor =0$
and $rhs=x_{ij}^{*}$, then the inequality follows. Overall, $\left(x_{ij}^{s},\, y_{i}^{s}\right)$
is a feasible solution.\end{proof}
\begin{lemma}
For the instances $\mathcal{I}^{l}$ and $\mathcal{I}^{s}$ the following
holds: \label{lem: ss-bounds}

\textup{(i)} $\max_{j\in\mathcal{C}}r_{j}^{l}\leq\sum_{i\in\mathcal{F}}R_{i}^{l}$
and $\max_{j\in\mathcal{C}}r_{j}^{s}\leq\sum_{i\in\mathcal{F}}R_{i}^{s}$.

\textup{(ii)} $R_{i}^{s}\in\left\{ 0,\,1,\,2\right\} $.\end{lemma}
\begin{proof}
(i) The previous lemma and the constraints of the LP \eqref{eq:ftra-lp}
together ensure the bounds that $\forall j\in\mathcal{C}:\, r_{j}^{l}\leq\sum_{i\in\mathcal{F}}x_{ij}^{l}\leq\sum_{i\in\mathcal{F}}y_{i}^{l}\leq\sum_{i\in\mathcal{F}}R_{i}^{l}$
and $r_{j}^{s}\leq\sum_{i\in\mathcal{F}}x_{ij}^{s}\leq\sum_{i\in\mathcal{F}}y_{i}^{s}\leq\sum_{i\in\mathcal{F}}R_{i}^{s}$.

(ii) We have $R_{i}^{s}=\mbox{\ensuremath{\left\lceil y_{i}^{*}\right\rceil }-}\max\left(0,\,\left\lfloor y_{i}^{*}\right\rfloor -1\right)$
after the substitution of $y_{i}^{l}$. If $\left\lfloor y_{i}^{*}\right\rfloor \geq1$,
then $R_{i}^{s}=\mbox{\ensuremath{\left\lceil y_{i}^{*}\right\rceil }-}\left\lfloor y_{i}^{*}\right\rfloor +1\in\left\{ 1,\,2\right\} $,
otherwise if $\left\lfloor y_{i}^{*}\right\rfloor <1$, then $R_{i}^{s}=\mbox{\ensuremath{\left\lceil y_{i}^{*}\right\rceil }}\in\left\{ 0,\,1\right\} $.\end{proof}
\begin{theorem}
If there is a $\rho$-approximation polynomial-time algorithm $\mathcal{A}$
for the general $FTRA$ with polynomially bounded $R_{i}$, which
always produces an integral solution that approximates the fractional
optimal solution with factor $\rho\geq1$. Then there is also a polynomial-time
$\rho$-approximation algorithm $\mathcal{A}'$ for the general $FTRA$.\label{reduction-theorem}\end{theorem}
\begin{proof}
We will describe such an algorithm $\mathcal{A}'$. It first does
the instance shrinking and splitting as described before for any instance
$\mathcal{I}$ of $FTRA$. From (i) of Lemma \ref{lem: ss-bounds},
the split instances $\mathcal{I}^{l}$ and $\mathcal{I}^{s}$ are
valid. From (ii), $\mathcal{I}^{s}$ has polynomially bounded $R_{i}^{s}$.
Note that if $R_{i}^{s}=0$, we can safely remove this site $i$ in
$\mathcal{I}^{s}$, and set the solution $\forall j\in\mathcal{C}:\, x_{ij}^{s},\, y_{i}^{s}=0$
when later combining it with $\mathcal{I}^{l}$. Then, $\mathcal{A}'$
uses $\mathcal{A}$ as a subroutine to solve $\mathcal{I}^{s}$ to
obtain a feasible integral solution $\left(\bar{x_{ij}^{s}},\,\bar{y_{i}^{s}}\right)$
that approximates the fractional optimum. Also, from Lemma \ref{lem:splitting},
$\left(\boldsymbol{x^{s}},\,\boldsymbol{y^{s}}\right)$ is feasible
to $\mathcal{I}^{s}$, so $cost\left(\boldsymbol{\bar{x^{s}}},\,\boldsymbol{\bar{y^{s}}}\right)\leq\rho\cdot cost\left(\boldsymbol{x^{s}},\,\boldsymbol{y^{s}}\right)$.
Finally, $\mathcal{A}'$ combines $\left(\bar{x_{ij}^{s}},\,\bar{y_{i}^{s}}\right)$
with the readily available constructed integer solution $\left(x_{ij}^{l},\, y_{i}^{l}\right)$
for $\mathcal{I}^{l}$. Because $\left(x_{ij}^{l},\, y_{i}^{l}\right)$
is a feasible integral solution to $\mathcal{I}^{l}$, then when combined,
they form a feasible integral solution to $\mathcal{I}$ as $r_{j}^{l}+r_{j}^{s}=r_{j}$
and $R_{i}^{l}+R_{i}^{s}=R_{i}=\left\lceil y_{i}^{*}\right\rceil $.
The only thing left is to prove the combined solution from $\mathcal{A}'$
is $\rho$-approximation, i.e., $cost\left(\boldsymbol{x^{l}},\,\boldsymbol{y^{l}}\right)+cost\left(\boldsymbol{\bar{x^{s}}},\,\boldsymbol{\bar{y^{s}}}\right)\leq\rho\cdot cost\left(\boldsymbol{x^{*}},\,\boldsymbol{y^{*}}\right)$.
This follows from $cost\left(\boldsymbol{x^{l}},\,\boldsymbol{y^{l}}\right)+cost\left(\boldsymbol{\bar{x^{s}}},\,\boldsymbol{\bar{y^{s}}}\right)\leq cost\left(\boldsymbol{x^{l}},\,\boldsymbol{y^{l}}\right)+\rho\cdot cost\left(\boldsymbol{x^{s}},\,\boldsymbol{y^{s}}\right)\leq\rho\cdot\left(cost\left(\boldsymbol{x^{l}},\,\boldsymbol{y^{l}}\right)+cost\left(\boldsymbol{x^{s}},\,\boldsymbol{y^{s}}\right)\right)$
and $\boldsymbol{x^{l}}+\boldsymbol{x^{s}}=\boldsymbol{x^{*}}$, $\boldsymbol{y^{l}}+\boldsymbol{y^{s}}=\boldsymbol{y^{*}}$.\end{proof}
\begin{corollary}
The general $FTRA$ is reducible to the general $FTFL$ in weakly
polynomial time.\end{corollary}
\begin{proof}
Any instance of $FTRA$ with polynomially bounded $R_{i}$ can be
treated as an equivalent $FTFL$ instance with facility size $\sum_{i\in\mathcal{F}}R_{i}$
which is polynomial. Then any polynomial time algorithm solves $FTFL$
with ratio $\rho$ (w.r.t. the fractional optimum) can become the
algorithm $\mathcal{A}$ for $\mathcal{A}'$ in the previous theorem
to solve $FTRA$ with the same ratio. In addition, the reduction requires
solving the LP first to obtain $\left(\boldsymbol{x^{*}},\,\boldsymbol{y^{*}}\right)$
which takes weakly polynomial time.
\end{proof}
Therefore, from the above corollary and the result of \cite{JaroslawFTFL1.725}
for the metric $FTFL$, we get the ratio of 1.7245 for the metric
$FTRA$. Also, \textcolor{black}{from the results of \cite{Jain00FTFL,Lin92filting},
we can deduce that the non-metric $FTFL$ has an approximation ratio
of $O\left(\log^{2}n\right)$. This is because Jain and Vazirani \cite{Jain00FTFL}
proved that $FTFL$ reduces to $UFL$ with a ratio loss of $O\left(\log n\right)$,
and Lin and Vitter \cite{Lin92filting} showed that the non-metric
$UFL$ can be approximated with the ratio of $O\left(\log n\right)$
w.r.t. the fractional optimum. For the non-metric $FTRA$, we can
subsequently achieve the same ratio due to its reduction to $FTFL$
first. Moreover, in future, any improved ratio for the general $FTFL$
might directly hold for the general $FTRA$.}

\section{The Uniform $FTRA$}

Interestingly, the reduction results in the previous section does
not imply that the uniform $FTRA$ reduces to the uniform $FTFL$
in weakly polynomial time. This is because the instance shrinking
technique may split a uniform instance into two general/non-uniform
instances. As a consequence, the ratio of 1.52 in \cite{Swamy08FTFL2.076}
for the uniform $FTFL$ does not directly hold for the uniform $FTRA$.
Nevertheless, in this section, we show this ratio can still be preserved
for the uniform $FTRA$ in strongly polynomial time with a primal-dual
algorithm and two acceleration heuristics. Note that the following
algorithms are generic which work for the general $FTRA$ as well.
The uniform condition is only necessary in the analysis (Lemma \ref{lem:tri}). 

We begin with a naive primal-dual (PD) algorithm \textcolor{black}{(Algorithm
2)} for $FTRA$ with an approximation ratio of 1.61 and then present
the first acceleration heuristic to improve the complexity of the
algorithm to strongly polynomial \textcolor{black}{$O\left(n^{4}\right)$.
W.l.o.g., the PD algorithm assumes that each client $j$ makes $r_{j}$
connections and each connection is associated with a }\textit{\textcolor{black}{port}}\textcolor{black}{{}
of $j$ denoted by }$j^{\left(q\right)}$\textcolor{black}{{} $\left(1\leq q\leq r_{j}\right)$.
Also, the function $\phi\left(j^{\left(q\right)}\right)$ represents
the facility/site a client $j$'s $q$-th port is connected with and
the variable $p_{j}$ keeps track of the port of the client $j$ to
be connected. The algorithm then gradually connects clients in the
port order from $1$ to $r_{j}$, as well as increasing the solution
$\left(\boldsymbol{x},\,\boldsymbol{y}\right)$ from $\left(\boldsymbol{0},\,\boldsymbol{0}\right)$
in its }\textit{\textcolor{black}{actions}}\textcolor{black}{{} in response
to some }\textit{\textcolor{black}{events }}\textcolor{black}{controlled
by a global time $t$ that increases monotonically from $0$. All
events }\textit{\textcolor{black}{repeatedly}}\textcolor{black}{{} occur
until all clients are fully-connected, i.e., the not-fully-connected
clients s}et \textcolor{black}{$\mathcal{U}=\emptyset$. At any $t$,
the }\textit{\textcolor{black}{payment}}\textcolor{black}{{} of any
client $j$ to a site $i$ is defined as $t$, and the }\textit{\textcolor{black}{contribution}}\textcolor{black}{{}
is $\max\left(0,\, t-c_{ij}\right)$ for the clients in $\mathcal{U}$
and $\max\left(0,\,\max_{q}c_{\phi\left(j^{\left(q\right)}\right)j}-c_{ij}\right)$
for the clients in $\mathcal{C}\backslash\mathcal{U}$. As $t$ increases,
the action that a client $j$ connects to a facility of $i$ ($x_{ij}$
is increased by $1$) happens under two events: Event 1. $j$'s payment
reaches the connection cost $c_{ij}$ of an already opened facility
at $i$ that $j$ is not connected to (implying at this time $y_{i}>x_{ij}$);
Event 2. sum of contributions of all clients to a closed facility
at $i$ reaches its opening cost $f_{i}$. In particular, if $y_{i}<R_{i}$,
Event 2 triggers the action that a new facility at $i$ is opened
first ($y_{i}$ is increased by 1). Then any client $j\in\mathcal{C}\backslash\mathcal{U}$
with }$\max_{q}c_{\phi\left(j^{\left(q\right)}\right)j}-c_{ij}>0$
will switch one of its most expensive connections from $\arg\max_{\phi\left(j^{\left(q\right)}\right)}c_{\phi\left(j^{\left(q\right)}\right)j}$
to $ $$i$; and the \textcolor{black}{client in $\mathcal{U}$ with
$t-c_{ij}\geq0$ will connect to $i$. In addition, for analyzing
the approximation ratio of PD, each port }$j^{\left(q\right)}$\textcolor{black}{{}
is associated with a dual variable $\alpha_{j}^{q}$ which is assigned
the time $t$ at which }$j^{\left(q\right)}$\textcolor{black}{{} gets
connected. From the algorithm, it should be obvious that $cost\left(\boldsymbol{x},\,\boldsymbol{y}\right)=\sum_{j\in\mathcal{C}}\sum_{1\leq q\leq r_{j}}\alpha_{j}^{q}$.
Note that Event 2 on $i$ stops occurring once $y_{i}=R_{i}$ and
this introduces some difficulties to the analysis. To tackle these
difficulties, we use an extra variable $\boldsymbol{\hat{x}}$ to
store the amounts of the clients' connections when they just become
fully-connected. Compared with the phase-connection approach in \cite{shihongftfa}
in which each phase constructs multiple stars connecting sites and
cities at the same port number in increasing order globally, our algorithm
constructs the most cost efficient star in increasing order of port
number confined within each star.}

\begin{algorithm}[H]
{\small \caption{PD: Primal-Dual Algorithm}
}{\small \par}

\textbf{\textcolor{black}{Input}}\textcolor{black}{: }$\mathcal{F},\,\mathcal{C},\,\boldsymbol{f},\,\boldsymbol{c},\,\boldsymbol{r},\,\boldsymbol{R}$.\textbf{\textcolor{black}{{}
Output: }}$\left(\boldsymbol{x},\,\boldsymbol{y}\right)$.

\textbf{\textcolor{black}{Initialization}}\textcolor{black}{: }Set
$\mathcal{U}=\mathcal{C}$, $\forall i\in\mathcal{F},\, j\in\mathcal{C}:\, x_{ij},\, y_{i}=0,\, p_{j}=1$.

\medskip{}

\textbf{\textcolor{black}{while}}\textcolor{black}{{} $\mathcal{U}\neq\emptyset$,
increase time $t$ uniformly and execute the events below:} 
\begin{itemize}
\item \textcolor{black}{Event 1: $\exists i\in\mathcal{F},\, j\in\mathcal{U}$:
$t=c_{ij}$ and $x_{ij}<y_{i}$}.\\
\textcolor{black}{Action 1: }\textbf{\textcolor{black}{set}}\textcolor{black}{{}
}$\phi\left(j^{\left(p_{j}\right)}\right)\leftarrow i$, $x_{ij}\leftarrow x_{ij}+1$
and $\alpha_{j}^{p_{j}}\leftarrow t$; If $p_{j}=r_{j}$, then\textbf{
$\mathcal{U}\leftarrow\mathcal{U}\backslash\left\{ j\right\} $} and
$\hat{x_{ij}}=x_{ij}$, otherwise $p_{j}\leftarrow p_{j}+1$.\textcolor{black}{{}
}\smallskip{}

\item \textcolor{black}{Event 2: $\exists i\in\mathcal{F}$: $\sum_{j\in\mathcal{U}}\max\left(0,\, t-c_{ij}\right)+$$\sum_{j\in\mathcal{C}\backslash\mathcal{U}}\max\left(0,\,\max_{q}c_{\phi\left(j^{\left(q\right)}\right)j}-c_{ij}\right)=f_{i}$}
and $y_{i}<R_{i}$.\\
\textcolor{black}{Action 2: }\textbf{\textcolor{black}{set}}\textcolor{black}{{}
$y_{i}\leftarrow y_{i}+1$}; $\forall j\in\mathcal{C}\backslash\mathcal{U}$
s.t. $\max_{q}c_{\phi\left(j^{\left(q\right)}\right)j}-c_{ij}>0:$
\textbf{set} $i_{j}^{*}\leftarrow\arg\max_{\phi\left(j^{\left(q\right)}\right)}c_{\phi\left(j^{\left(q\right)}\right)j}$,
$x_{i_{j}^{*}j}\leftarrow x_{i_{j}^{*}j}-1$, $x_{ij}\leftarrow x_{ij}+1$
and $\phi\left(j^{\left(\arg\max_{q}c_{\phi\left(j^{\left(q\right)}\right)j}\right)}\right)\leftarrow i$;
$\forall j\in\mathcal{\mathcal{U}}$ s.t. $t\geq c_{ij}:$ \textbf{do}
Action 1.\end{itemize}
\end{algorithm}

\begin{rem} \textcolor{black}{If more than one event happen at time
$t$, the PD algorithm processes all of them in an arbitrary order.}
Also, the events themselves may repeatedly happen at any $t$ because
more than one facilities at a site are allowed to open.\end{rem}
\begin{lemma}
Algorithm PD computes a feasible solution to the uniform $FTRA$ and
runs in $O\left(n^{3}\max_{j\in\mathcal{C}}r_{j}\right)$.\label{lem:time-PD}\end{lemma}
\begin{proof}
The solution is feasible because $\left(\boldsymbol{x},\,\boldsymbol{y}\right)$
produced from PD is feasible to LP \eqref{eq:ftra-ip}. Each iteration
of PD at least connects a port of a client, so there are maximum $\sum_{j\in\mathcal{C}}r_{j}$
iterations. In addition, \textcolor{black}{similar to Theorem 22.4
of \cite{JensVygenFL06Book} and Theorem 8 of \cite{jain01approximation}
for $UFL$, }the client switching in Action 2 dominates\textcolor{black}{{}
the time complexity. In each iteration, the switching takes time $O\left(n_{c}n_{f}\right)$
to update clients' contributions to other facilities for computing
the anticipated times of the events. Hence, the total time is $O\left(\sum_{j\in\mathcal{C}}r_{j}n_{c}n_{f}\right)$,
i.e. $O\left(n^{3}\max_{j\in\mathcal{C}}r_{j}\right)$.\medskip{}
}
\end{proof}
The PD algorithm produces a feasible primal solution $\left(\boldsymbol{x},\,\boldsymbol{y}\right)$
to $FTRA$ but an infeasible dual solution $\left(\boldsymbol{\alpha},\boldsymbol{\beta},\,\boldsymbol{z}\right)$
if we simply let $\alpha_{j}=\alpha_{j}^{r_{j}}$, $\beta_{ij}=\max\left(0,\alpha_{j}-c_{ij}\right)$
and $z_{i}=0$ in LP \eqref{eq:ftra-dual}. This is because although
the LP's second constraint holds, the first constraint fails to hold
since the algorithm only guarantees \textcolor{black}{$\forall i\in\mathcal{F}$:
$\sum_{j\in\mathcal{U}}\max\left(0,\, t-c_{ij}\right)$$\leq f_{i}$
where $\mathcal{U}\subseteq{\cal C}$.} In order to get around this
feasibility issue, there are ways like the classical dual fitting
\cite{Jain03dualfitting} and inverse dual fitting \cite{shihongftfa}
analyses. However, we observe that these approaches are actually both
based on constraints of the LP. Therefore, in the following we develop
a simple and step-by-step constraint-based analysis for the ease of
handling more complicated dual constructions. Together with the factor-revealing
technique of \cite{Jain03dualfitting}, we derive the ratio of 1.61
for $FTRA$.

In the analysis, first, we use the following dual constructions of
$\boldsymbol{\alpha},\,\boldsymbol{z}$ to bound the primal solution
cost $cost\left(\boldsymbol{x},\,\boldsymbol{y}\right)$ with the
dual solution cost $cost\left(\boldsymbol{\alpha},\,\boldsymbol{\beta},\,\boldsymbol{z}\right)$.

$\forall i\in\mathcal{F},\, j\in\mathcal{C}:\,\alpha_{j}=\alpha_{j}^{r_{j}},\, z_{i}=\sum_{j\in\mathcal{C}}\pi_{ij}$
where $\pi_{ij}=\begin{cases}
\frac{\hat{x_{ij}}\left(\alpha_{j}-\alpha_{j}^{l_{ij}}\right)}{R_{i}} & \textrm{if}\,\hat{x_{ij}}=R_{i}\\
0 & \textrm{if}\,\hat{x_{ij}}<R_{i}
\end{cases}$.

\medskip{}

In this setting, \textcolor{black}{$\boldsymbol{\hat{x}}$} stores
the primary connection amounts of the clients after they become fully-connected
but before they switch any of their connections. $\forall i\in\mathcal{F},\, j\in\mathcal{C}:$
$l_{ij}$ denotes the last port of $j$ \textit{connecting to $i$}
before switching, so $\alpha_{j}^{l_{ij}}$ is the dual value of the
port $l_{ij}$; $\alpha_{j}^{r_{j}}$ is the dual of $j$'s last port
and $\pi_{ij}$ can be interpreted as the conditional marginal dual/price
depending on whether $\hat{x_{ij}}$ reaches $R_{i}$ or not. $\forall i\in\mathcal{F}:$
$z_{i}$ is therefore the sum of marginal duals of all clients w.r.t.
$i$. The other dual variable $\boldsymbol{\beta}$ is to be constructed
later in the analysis.
\begin{lemma}
$cost\left(\boldsymbol{x},\,\boldsymbol{y}\right)\leq cost\left(\boldsymbol{\alpha},\,\boldsymbol{\beta},\,\boldsymbol{z}\right)$
where $\left(\boldsymbol{x},\,\boldsymbol{y}\right)$ is the feasible
primal solution produced from the PD algorithm and $\left(\boldsymbol{\alpha},\,\boldsymbol{\beta},\,\boldsymbol{z}\right)$
is constructed from above. \label{lem: pdb}\end{lemma}
\begin{proof}
{\small 
\begin{eqnarray*}
 & cost\left(\boldsymbol{\alpha},\,\boldsymbol{\beta},\,\boldsymbol{z}\right) & =\sum_{j\in\mathcal{C}}r_{j}\alpha_{j}-\sum_{i\in\mathcal{F}}R_{i}z_{i}\\
 &  & =\sum_{j\in\mathcal{C}}\sum_{i\in\mathcal{F}}x_{ij}\alpha_{j}^{r_{j}}-\underset{\textrm{if}\,\hat{x_{ij}}=R_{i}}{\underbrace{\sum_{i\in\mathcal{F}}\sum_{j\in\mathcal{C}}\hat{x_{ij}}\left(\alpha_{j}^{r_{j}}-\alpha_{j}^{l_{ij}}\right)}}-\underset{\textrm{if}\,\hat{x_{ij}}<R_{i}}{\underbrace{0}}\\
 &  & =\sum_{j\in\mathcal{C}}\sum_{i\in\mathcal{F}}\hat{x_{ij}}\alpha_{j}^{r_{j}}-\underset{\textrm{if}\,\hat{x_{ij}}=R_{i}}{\underbrace{\sum_{i\in\mathcal{F}}\sum_{j\in\mathcal{C}}\hat{x_{ij}}\left(\alpha_{j}^{r_{j}}-\alpha_{j}^{l_{ij}}\right)}}\\
 &  & =\underset{\textrm{if}\,\hat{x_{ij}}<R_{i}}{\underbrace{\sum_{i\in\mathcal{F}}\sum_{j\in\mathcal{C}}\hat{x_{ij}}\alpha_{j}^{r_{j}}}}+\underset{\textrm{if}\,\hat{x_{ij}}=R_{i}}{\underbrace{\sum_{i\in\mathcal{F}}\sum_{j\in\mathcal{C}}\hat{x_{ij}}\alpha_{j}^{l_{ij}}}}.
\end{eqnarray*}
}{\small \par}

Hence, $cost\left(\boldsymbol{x},\,\boldsymbol{y}\right)=\sum_{j\in\mathcal{C}}\sum_{1\leq q\leq r_{j}}\alpha_{j}^{q}\leq\underset{\textrm{if}\,\hat{x_{ij}}<R_{i}}{\underbrace{\sum_{i\in\mathcal{F}}\sum_{j\in\mathcal{C}}\hat{x_{ij}}\alpha_{j}^{r_{j}}}}+\underset{\textrm{if}\,\hat{x_{ij}}=R_{i}}{\underbrace{\sum_{i\in\mathcal{F}}\sum_{j\in\mathcal{C}}\hat{x_{ij}}\alpha_{j}^{l_{ij}}}}=cost\left(\boldsymbol{\alpha},\,\boldsymbol{\beta},\,\boldsymbol{z}\right)$.
\end{proof}
Second, we exploit the dual constraints of the LP \eqref{eq:ftra-dual}
by relaxing their feasibilities with some relaxation factors. Before
going into this, we have the basic definition below.
\begin{definition}
An algorithm is bi-factor $\left(\rho_{f},\,\rho_{c}\right)$ or single
factor $\max\left(\rho_{f},\,\rho_{c}\right)$-approximation for $FTRA$,
iff for every instance $\mathcal{I}$ of $FTRA$ and any feasible
solution $SOL$ (possibly fractional) of $\mathcal{I}$ with facility
cost $F_{SOL}$ and connection cost $C_{SOL},$ the total cost produced
from the algorithm is at most $\rho_{f}F_{SOL}+\rho_{c}C_{SOL}$ ($\rho_{f},\,\rho_{c}$
are both positive constants greater than or equal to one).
\end{definition}
In the definition, let any feasible solution be $SOL=\left(\boldsymbol{x''},\,\boldsymbol{y''}\right)$,
then $F_{SOL}=\sum_{i\in\mathcal{F}}f_{i}y_{i}''$, $C_{SOL}=\sum_{i\in\mathcal{F}}\sum_{j\in\mathcal{C}}c_{ij}x_{ij}''$
and $cost\left(\boldsymbol{x''},\,\boldsymbol{y''}\right)=F_{SOL}+C_{SOL}$.
In the following, we consider the feasibility relaxed dual constraints
with the relaxation factors $\rho_{f}$ and $\rho_{c}$:

\medskip{}

(C6) $\forall i\in\mathcal{F},\, j\in\mathcal{C}:\,\alpha_{j}-\beta_{ij}\leq\rho_{c}c_{ij}$.

(C7) $\forall i\in\mathcal{F}:\,\sum_{j\in\mathcal{C}}\beta_{ij}\leq\rho_{f}f_{i}+z_{i}$.

\medskip{}

Next, we show that if the dual variables $\left(\boldsymbol{\alpha},\boldsymbol{\beta},\,\boldsymbol{z}\right)$
satisfies these relaxed constraints, the corresponding dual cost will
be bounded by any feasible primal cost scaled by the factors $\rho_{f}$
and $\rho_{c}$.
\begin{lemma}
If $\left(\boldsymbol{\alpha},\boldsymbol{\beta},\,\boldsymbol{z}\right)$
satisfies (C6) and (C7) while $SOL=\left(\boldsymbol{x''},\,\boldsymbol{y''}\right)$
is any feasible primal solution, then $cost\left(\boldsymbol{\alpha},\boldsymbol{\beta},\,\boldsymbol{z}\right)\leq\rho_{f}F_{SOL}+\rho_{c}C_{SOL}$.\label{lem:dual-bound}\end{lemma}
\begin{proof}
Since $\left(\boldsymbol{x''},\,\boldsymbol{y''}\right)$ is any feasible
solution, all constraints of the LP \eqref{eq:ftra-lp} should hold
first. Together with (C6) and (C7), we have:

{\small 
\begin{eqnarray*}
 & cost\left(\boldsymbol{\alpha},\boldsymbol{\beta},\,\boldsymbol{z}\right) & =\sum_{j\in\mathcal{C}}r_{j}\alpha_{j}-\sum_{i\in\mathcal{F}}R_{i}z_{i}\\
 &  & \leq\sum_{j\in\mathcal{C}}\sum_{i\in\mathcal{F}}x_{ij}''\alpha_{j}-\sum_{i\in\mathcal{F}}y_{i}''z_{i}\\
 &  & \leq\sum_{i\in\mathcal{F}}\sum_{j\in\mathcal{C}}\left[\beta_{ij}y_{i}''+\left(\alpha_{j}-\beta_{ij}\right)x_{ij}''\right]-\sum_{i\in\mathcal{F}}y_{i}''z_{i}\\
 &  & \leq\sum_{i\in\mathcal{F}}\left(\rho_{f}f_{i}+z_{i}\right)y_{i}''+\sum_{i\in\mathcal{F}}\sum_{j\in\mathcal{C}}\rho_{c}c_{ij}x_{ij}''-\sum_{i\in\mathcal{F}}y_{i}''z_{i}\\
 &  & =\sum_{i\in\mathcal{F}}\rho_{f}f_{i}y_{i}''+\sum_{i\in\mathcal{F}}\sum_{j\in\mathcal{C}}\rho_{c}c_{ij}x_{ij}''=\rho_{f}F_{SOL}+\rho_{c}C_{SOL}.
\end{eqnarray*}
}{\small \par}
\end{proof}
The previous two lemmas and the definition immediately imply the next
lemma.
\begin{lemma}
The PD Algorithm is $\left(\rho_{f},\,\rho_{c}\right)$-approximation
if $\left(\boldsymbol{\alpha},\,\boldsymbol{\beta},\,\boldsymbol{z}\right)$
satisfies (C6) and (C7).\label{lem:appro}
\end{lemma}
In the last step, we show $\left(\boldsymbol{\alpha},\,\boldsymbol{\beta},\,\boldsymbol{z}\right)$
indeed satisfies (C6) and (C7), so the algorithm is $\left(\rho_{f},\,\rho_{c}\right)$-approximation.
To satisfy (C6), obviously we can set $\forall i\in\mathcal{F},\, j\in\mathcal{C}:\,\beta_{ij}=\max\left(0,\alpha_{j}-\rho_{c}c_{ij}\right)$
(because $\beta_{ij}\geq0$), thereby finishing constructing $\left(\boldsymbol{\alpha},\,\boldsymbol{\beta},\,\boldsymbol{z}\right)$.
The rest is to find the actual values of the factors $\rho_{f}$ and
$\rho_{c}$ to make (C7) hold as well. The next lemma and corollary
are more specific forms of the previous lemma, after substituting
the setting of $\left(\boldsymbol{\alpha},\,\boldsymbol{\beta},\,\boldsymbol{z}\right)$
into (C7).
\begin{lemma}
The PD Algorithm is $\left(\rho_{f},\,\rho_{c}\right)$-approximation
if $\forall i\in\mathcal{F}:\,$\linebreak{}
$\sum_{j\in\mathcal{A}_{i}}\left(\alpha_{j}^{r_{j}}-\rho_{c}c_{ij}-\pi_{ij}\right)\leq\rho_{f}f_{i}$
where $\mathcal{A}_{i}=\left\{ j\in\mathcal{C}\,|\,\alpha_{j}^{r_{j}}\geq\rho_{c}c_{ij}\right\} $.\label{lem:s-appro}\end{lemma}
\begin{proof}
After the substitution, (C7) becomes:

$\forall i\in\mathcal{F}:\,\sum_{j\in\mathcal{C}}\left(\beta_{ij}-\pi_{ij}\right)\leq\rho_{f}f_{i}\Rightarrow$

$\forall i\in\mathcal{F}:\,\sum_{j\in\mathcal{C}}\left(\max\left(0,\,\alpha_{j}^{r_{j}}-\rho_{c}c_{ij}\right)-\pi_{ij}\right)\leq\rho_{f}f_{i}\Rightarrow$

$\forall i\in\mathcal{F}:\,\underset{\textrm{if}\,\alpha_{j}^{r_{j}}\geq\rho_{c}c_{ij}}{\underbrace{\sum_{j\in\mathcal{C}}\left(\alpha_{j}^{r_{j}}-\pi_{ij}-\rho_{c}c_{ij}\right)}}-\underset{\textrm{if}\,\alpha_{j}^{r_{j}}<\rho_{c}c_{ij}}{\underbrace{\sum_{j\in\mathcal{C}}\pi_{ij}}}\leq\rho_{f}f_{i}$.

Therefore, since $\pi_{ij}\geq0$, it is sufficient to prove $\forall i\in\mathcal{F}:\,\sum_{j\in\mathcal{A}_{i}}\left(\alpha_{j}^{r_{j}}-\pi_{ij}-\rho_{c}c_{ij}\right)\leq\rho_{f}f_{i}$
where $\mathcal{A}_{i}=\left\{ j\in\mathcal{C}\,|\,\alpha_{j}^{r_{j}}\geq\rho_{c}c_{ij}\right\} $
to satisfy the original (C7).
\end{proof}
If we set $\forall i\in\mathcal{F},\, j\in\mathcal{C}:\, u_{ij}=\alpha_{j}^{r_{j}}-\pi_{ij}$,
then $u_{ij}=\begin{cases}
\alpha_{j}^{l_{ij}} & \textrm{if}\,\hat{x_{ij}}=R_{i}\\
\alpha_{j}^{r_{j}} & \textrm{if}\,\hat{x_{ij}}<R_{i}
\end{cases}$ and we have the corollary below.
\begin{corollary}
W.l.o.g., for every site $i$, order the corresponding $n_{i}=\left|\mathcal{A}_{i}\right|$
clients in $\mathcal{A}_{i}=\left\{ j\in\mathcal{C}\,|\,\alpha_{j}^{r_{j}}\geq\rho_{c}c_{ij}\right\} $
s.t. $u_{i1}\leq\ldots\leq u_{in_{i}}$. Then the PD Algorithm is
$\left(\rho_{f},\,\rho_{c}\right)$-approximation if $\forall i\in\mathcal{F}:\,\sum_{j=1}^{n_{i}}\left(u_{ij}-\rho_{c}c_{ij}\right)\le\rho_{f}f_{i}$.
\label{cor:s-appro}
\end{corollary}
In addition, for each $i$, any subset of the clients are ordered
from $1$ to $k_{i}$ s.t. $u_{i1}\leq\ldots\leq u_{ik_{i}}$. Now,
we proceed the proof to find $\rho_{f}$ and $\rho_{c}$ with the
following lemmas. These lemmas are needed for the factor-revealing
technique and they capture the properties of the PD algorithm for
the uniform $FTRA$.
\begin{lemma}
For every site $i$, at time $t=u_{ij}-\epsilon,\,$ $\forall1\leq h<j<k_{i}$
let $\omega_{h,\, j}^{i}=\begin{cases}
u_{ih} & if\,\hat{x_{ih}}=R_{i}\\
\max_{q}c_{\phi\left(h^{\left(q\right)}\right)h} & if\,\hat{x_{ih}}<R_{i}
\end{cases}$, then $\,\omega_{h,j}^{i}\geq\omega_{h,j+1}^{i}$.\label{lem:r}\end{lemma}
\begin{proof}
If $\hat{x_{ih}}=R_{i}$, then $\omega_{h,\, j}^{i}=\omega_{h,j+1}^{i}$
(at time $t=u_{i\left(j+1\right)}-\epsilon$) $=u_{ih}$. Otherwise,
$\hat{x_{ih}}<R_{i}$ implies $u_{ih}=\alpha_{h}^{r_{h}}$, so $h$
is fully-connected at time $t$ since $ $$u_{ih}\leq u_{ij}$. Therefore,
$\omega_{h,j}^{i}\geq\omega_{h,j+1}^{i}$ because a fully-connected
client's ports always reconnect to the sites with less connection
cost, so its maximum connection cost will never increase. The lemma
follows.\end{proof}
\begin{lemma}
For any site $i$ and ordered $k_{i}$ clients, $\forall1\leq j\leq k_{i}:\,$\linebreak{}
$\sum_{h=1}^{j-1}\max\left(0,\,\omega_{h,j}^{i}-c_{ih}\right)+\sum_{h=j}^{k_{i}}\max\left(0,\, u_{ij}-c_{ih}\right)\leq f_{i}$.\label{lem:contri}\end{lemma}
\begin{proof}
For any site $i$ and at time $t=u_{ij}-\epsilon$, if $h<j$ client
$h$'s contribution is set to be $\max\left(0,\,\omega_{h,j}^{i}-c_{ih}\right)$.
In particular, from the previous lemma and the setting of $u_{ij}$,
if $\hat{x_{ih}}<R_{i}$, it implies $h$ is fully-connected at time
$t$ and the contribution is $\max\left(0,\,\max_{q}c_{\phi\left(h^{\left(q\right)}\right)h}-c_{ih}\right)$.
In addition, if $\hat{x_{ih}}=R_{i}$ the contribution is $\max\left(0,\,\alpha_{h}^{l_{ih}}-c_{ih}\right)$.
Note that under this case, $h$ still might be fully-connected at
time $t$, but because $\hat{x_{ih}}=R_{i}$ and following the algorithm,
its contribution should not be set to $\max\left(0,\,\max_{q}c_{\phi\left(h^{\left(q\right)}\right)h}-c_{ih}\right)$
for ensuring the lemma. On the other hand, if $h\geq j$, $h$ is
not fully-connected since $t<\alpha_{h}^{r_{h}}$, so we set the contribution
to $\max\left(0,\, t-c_{ih}\right)$, i.e. $\max\left(0,\, u_{ij}-c_{ih}\right)$.
From the execution of the algorithm, at any time, the sum of these
contributions will not exceed the facility's opening cost at site
$i$, hence the lemma follows.\end{proof}
\begin{lemma}
For any site $i$ and clients $h,\, j$ s.t. $1\leq h<j\leq k_{i}:$
$r_{h}=r_{j}=r$, then ${\displaystyle u_{ij}\leq\omega_{h,j}^{i}+c_{ij}}+c_{ih}$.\label{lem:tri}\end{lemma}
\begin{proof}
At time $t=u_{ij}-\epsilon,$ if all facilities at site $i$ are already
open, then $u_{ij}\leq c_{ij}$ and the lemmas holds. Otherwise, if
not all facilities are open, then at time $t$ every client $h<j$
is fully-connected. This is because $u_{ih}\leq u_{ij}$ implies $u_{ih}=\alpha_{h}^{l_{ih}}$
or $\alpha_{h}^{r_{h}}$ at the time $t$. Since $h$ can only connect
to less than $R_{i}$ facilities at $i$, this contradicts the condition
$\hat{x_{ih}}=R_{i}$ for the setting of $u_{ih}$, so $u_{ih}=\alpha_{h}^{r_{h}}.$
In addition, $j$ itself is not fully-connected at $t$, whereas $h$
is fully-connected and has already connected to $r$ facilities. There
is at least a facility to which $h$ is connected but not $j$. (This
is where we must enforce all clients have the uniform connection $r$.)
Denote this facility (site) by $i'$, we have $u_{ij}\leq c_{i'j}$
and $\omega_{h,j}^{i}\geq c_{i'h}$. Lastly, by the triangle inequality
of the metric property, $c_{i'j}\leq c_{i'h}+c_{ij}+c_{ih}$ and then
we have the lemma.
\end{proof}
{\small 
\begin{align*}
z_{k}=\textrm{maximize\,\,\,\,\,} & {\displaystyle \frac{\sum_{j=1}^{k}\alpha_{j}}{f+\sum_{j=1}^{k}d_{j}}}\\
\textrm{subject to\,\,\,\,\,} & \forall1\leq j<k:\alpha_{j}\leq\alpha_{j+1}\\
 & \forall1\leq h<j<k:r_{h,j}\geq r_{h,j+1}\\
 & \forall1\leq h<j\leq k:{\displaystyle \alpha_{j}\leq r_{h,j}+d_{h}+d_{j}}\\
 & 1\leq j\leq k:\sum_{h=1}^{j-1}\max\left(r_{h,j}-d_{h},\,0\right)+\sum_{h=j}^{k}\max\left(\alpha_{j}-d_{h},\,0\right)\leq f\\
 & 1\leq h\leq j<k:\alpha_{j},\, d_{j},\, f,\, r_{h,j}\geq0
\end{align*}
}Consider the above factor-revealing program series (25) of \cite{Jain03dualfitting}.
If we let $k=k_{i,\,}\alpha_{j}=u_{ij},r_{h,j}=\omega_{h,j}^{i},\, f=f_{i},\, d_{j}=c_{ij}$,
from the previous lemmas it is clear that $u_{ij},\,\omega_{h,j}^{i},\, f_{i}$
and $c_{ij}$ constitute a feasible solution. Also, from Lemma 5.4
and Theorem 8.3 of \cite{Jain03dualfitting}, and Lemma 4 of \cite{Mohammad06FLP}
we can directly get $\forall i\in\mathcal{F}:\,$$\sum_{j=1}^{k_{i}}\left(u_{ij}-1.61c_{ij}\right)\leq1.61f_{i},\,\sum_{j=1}^{k_{i}}\left(u_{ij}-1.78c_{ij}\right)\leq1.11f_{i}$
and $\sum_{j=1}^{k_{i}}\left(u_{ij}-2c_{ij}\right)\leq f_{i}$. Furthermore,
because $n_{i}=\left|\mathcal{A}_{i}\right|$ and $k_{i}$ represents
the size of any subset of the clients, Lemma \ref{lem:time-PD} and
Corollary \ref{cor:s-appro} directly lead to the following theorem.
\begin{theorem}
Algorithm PD is 1.61-, (1.11, 1.78)- and (1,2)-approximation in time
$O\left(n^{3}\max_{j\in\mathcal{C}}r_{j}\right)$ for the uniform
$FTRA$. \label{thm: multifactor}
\end{theorem}
\begin{algorithm}[H]
\caption{APD: Acceleration of Primal-Dual Algorithm}

\textbf{\textcolor{black}{Input}}\textcolor{black}{: }$\mathcal{F},\,\mathcal{C},\,\boldsymbol{f},\,\boldsymbol{c},\,\boldsymbol{r},\,\boldsymbol{R}$.\textbf{\textcolor{black}{{}
Output: }}$\left(\boldsymbol{x},\,\boldsymbol{y}\right)$.

\textbf{\textcolor{black}{Initialization}}\textcolor{black}{: S}et
$\mathcal{U}=\mathcal{C}$, $\forall i\in\mathcal{F},\, j\in\mathcal{C}:\, x_{ij},\, y_{i}=0,\, FC_{j}=0$.

\medskip{}

\textbf{\textcolor{black}{while}}\textcolor{black}{{} $\mathcal{U}\neq\emptyset$,
increase time $t$ uniformly and execute the events below:} 
\begin{itemize}
\item \textcolor{black}{Event 1: $\exists i\in\mathcal{F},\, j\in\mathcal{U}$
s.t. $t=c_{ij}$ and }$x_{ij}<y_{i}$.\\
\textcolor{black}{Action 1-a: }$ToC\leftarrow\min\left(y_{i}-x_{ij},\, r_{j}-FC_{j}\right)$;\\
Action 1-b: \textbf{set} $x_{ij}\leftarrow x_{ij}+ToC$ and $FC_{j}\leftarrow FC_{j}+ToC$;\\
Action 1-c: If $FC_{j}=r_{j}$ then\textbf{ $\mathcal{U}\leftarrow\mathcal{U}\backslash\left\{ j\right\} $}.\textcolor{black}{{}
}\medskip{}

\item \textcolor{black}{Event 2:} \textcolor{black}{$\exists i\in\mathcal{F}$:
$\sum_{j\in\mathcal{U}}\max\left(0,\, t-c_{ij}\right)+$$\sum_{j\in\mathcal{C}\backslash\mathcal{U}}\max\left(0,\,\max_{i'\in\mathcal{F}\&\&x_{i'j}>0}c_{i'j}-c_{ij}\right)=f_{i}$}
and $y_{i}<R_{i}$.\\
\textcolor{black}{Action 2-a: }$\mathcal{\mathcal{U}}_{i}\leftarrow\left\{ j\in\mathcal{\mathcal{U}}\,|\, t-c_{ij}\geq0\right\} $
and $NC\leftarrow\min_{j\in\mathcal{\mathcal{U}}_{i}}\left(r_{j}-FC_{j}\right)$;
\\
Action 2-b: $\mathcal{S}_{i}\leftarrow\left\{ j\in\mathcal{\mathcal{\mathcal{C}\backslash\mathcal{U}}}\,|\,\max_{i'\in\mathcal{F}\&\&x_{i'j}>0}c_{i'j}-c_{ij}>0\right\} $,
$\forall j\in\mathcal{S}_{i}:\, i_{j}^{*}\leftarrow\arg\max_{i'\in\mathcal{F}\&\&x_{i'j}>0}c_{ij}$
and $NS\leftarrow\min_{j\in\mathcal{S}_{i}}x_{i_{j}^{*}j}$; \\
Action 2-c: \textbf{\textcolor{black}{set}}\textcolor{black}{{} }$ToC\leftarrow\min\left(NC,\, NS,\, R_{i}-y_{i}\right)$
and $y_{i}\leftarrow y_{i}+ToC$; \\
Action 2-d: $\forall j\in\mathcal{S}_{i}:$ $x_{i_{j}^{*}j}\leftarrow x_{i_{j}^{*}j}-ToC$
and $x_{ij}\leftarrow x_{ij}+ToC$; \\
Action 2-e: $\forall j\in\mathcal{\mathcal{U}}_{i}:$ \textbf{do}
Action 1-b; \\
Action 2-f: $\forall j\in\mathcal{\mathcal{U}}_{i}:$ \textbf{do}
Action 1-c.
\end{itemize}
\begin{rem} \textcolor{black}{For convenience of analysis, sequential
actions of the events are separated as above. If more than one event
happen at the same time, the algorithm process Event 2 first so that
no repeated events are needed.}\end{rem}
\end{algorithm}

The previous PD algorithm runs in pseudo-polynomial time depending
on $r_{j}$. With the acceleration heuristic described in the following,
the algorithm can then change to an essentially identical algorithm
APD (Algorithm 3) which is strongly polynomial. In the heuristic,\textcolor{black}{{}
$\left(\boldsymbol{x},\,\boldsymbol{y}\right)$ is able to increase
at a }faster rate\textcolor{black}{{} }rather than $1$, through combining
the repeated events into a single event to reduce the total number
of events to process and hence achieve fast connections. In particular,
for \textcolor{black}{Event }2, once a facility of a site $i$ is
opened and connected with a group of clients' ports, according to
the PD algorithm, additional facilities at $i$ may subsequently open
and connect with this group of clients' other ports until their sum
of contributions (SOC) becomes insufficient to pay $f_{i}$, or $y_{i}=R_{i}$.
The SOC is not enough \textcolor{black}{any more if} a client in $\mathcal{U}$
appears to be fully-connected, so $\sum_{j\in\mathcal{U}}\max\left(0,\, t-c_{ij}\right)$
will decrease, or the most expensive connection of a client in $\mathcal{C}\backslash\mathcal{U}$
differs (after switching all such connections), in this case $\sum_{j\in\mathcal{C}\backslash\mathcal{U}}\max\left(0,\,\max_{i'\in\mathcal{F}\&\&x_{i'j}>0}c_{i'j}-c_{ij}\right)$%
\footnote{For simplicity of the algorithm description, we replace the term \textcolor{black}{$\sum_{j\in\mathcal{C}\backslash\mathcal{U}}\max\left(0,\,\max_{q}c_{\phi\left(j^{\left(q\right)}\right)j}-c_{ij}\right)$}
in the PD algorithm with essentially the same term here.%
} will decrease. Similarly, for \textcolor{black}{Event }1, once a
client $j$'s port starts to connect to an already opened facility
at a site $i$, its other ports may get connected to $i$ at the same
time until either there are no remaining open facilities at $i$ or
$j$ reaches $r_{j}$ connections. 

Formally in the APD Algorithm, $FC_{j}$ denotes the number of established
connections of client $j$ and $ToC$ the total number of connections
decided to make according to the heuristic. The incremental rate of
\textcolor{black}{$\left(\boldsymbol{x},\,\boldsymbol{y}\right)$}
can then be determined by $ToC$ instead of $1$. Moreover, in the
more complicated Event 2 on a site $i$, $NC$ denotes the maximum
number of connections to make until one of the clients in $\mathcal{U}$
gets fully-connected, and $NS$ the maximum number of switches until
the most expensive connection of a client in $\mathcal{C}\backslash\mathcal{U}$
changes. Therefore, $ToC$ is calculated as $\min\left(NC,\, NS,\, R_{i}-y_{i}\right)$,
the maximum number of connections until the SOC becomes insufficient
or $y_{i}=R_{i}$. Similarly, for \textcolor{black}{Event }1, $ToC$
is calculated as $\min\left(y_{i}-x_{ij},\, r_{j}-FC_{j}\right)$.
\begin{lemma}
With the acceleration heuristic decided by \textup{$ToC$}, the numbers
of Event 1 and Event 2 in the APD algorithm are bounded by \textcolor{black}{$n_{f}n_{c}$}
and \textcolor{black}{$\left(n_{c}+n_{f}+n_{c}n_{f}\right)$} respectively
which are independent of $r_{j}$.\label{lem:suh}\end{lemma}
\begin{proof}
In the APD algorithm, the number of Event 1 is \textcolor{black}{at
most $n_{f}n_{c}$ because for any client $j$ and site $i$ only
when $t=c_{ij}$, $j$ exhaustively gets connected to open facilities
at site $i$, and there are $n_{f}$ sites and $n_{c}$ clients in
total. }Moreover, \textcolor{black}{each Event 2 will cause at least
one of the following 3 cases: (1) a client $j$ in }$\mathcal{\mathcal{U}}$\textcolor{black}{{}
becomes fully-connected; (2) a client $j$ in }$\mathcal{\mathcal{\mathcal{C}\backslash\mathcal{U}}}$\textcolor{black}{{}
switches all of its most expensive connections; (3) a site opens all
its facilities. It is easy to see that there are maximum $n_{c}$
and $n_{f}$ cases (1) and (3) respectively, so we are left to bound
the number of case (2). For this case, it is important to observe
that any client $j$ has at most $n_{f}$ possible sets of connections
where connections in each set associate to the same site. So there
are at most $n_{c}n_{f}$ such possible sets in total, and each case
(2) removes at least one set of a client with currently most expensive
connection cost, effectively reducing the number of possible sites
for switching, since clients only switch to cheaper connections. Therefore,
there are at most $n_{c}n_{f}$ case (2) and Event 2 is bounded by
$\left(n_{c}+n_{f}+n_{c}n_{f}\right)$.}\end{proof}
\begin{lemma}
Algorithm APD computes a feasible solution to the $FTRA$ and runs
in $O\left(n^{4}\right)$.\label{lem:time-SPD}\end{lemma}
\begin{proof}
The solution is feasible because APD is essentially the same as the
PD algorithm for $FTRA$ except the implementation of the acceleration
heuristic. In addition,\textcolor{black}{{} }from the previous lemma,
the number of Event 1 is \textcolor{black}{at most $n_{f}n_{c}$,
so the numbers of both Action 1-a and 1-b are bounded by $n_{f}n_{c}$,
while Action 1-c is bounded by $n_{c}$ since there are $n_{c}$ clients
to be connected in total. In addition, the number of Event 2 is bounded
by $\left(n_{c}+n_{f}+n_{c}n_{f}\right)$, so the numbers of Action
2-a, 2-b, 2-c, 2-d and 2-e are bounded by $O\left(n_{c}n_{f}\right)$
while Action 2-f is included in Action 1-c. Although the presented
APD algorithm is continuous on $t$, in real implementation, it can
be easily discretized through finding the the smallest $t$ that satisfy
the conditions of events. Naively, finding such $t$ for Event 1 takes
time $O\left(n_{c}n_{f}\right)$, and Event 2 takes time $O\left(n_{c}n_{f}^{2}\right)$,
so the algorithm runs in }$O\left(n^{5}\right)$. However, \textcolor{black}{similar
to Theorem 22.4 of \cite{JensVygenFL06Book} and Theorem 8 of \cite{jain01approximation}
for $UFL$, we can maintain two heaps to reduce the time to find $t$.
In particular,} each Action 2-d \textcolor{black}{actually requires
extra time $O\left(n_{c}n_{f}\right)$ after }the client switchings\textcolor{black}{{}
to change clients' contributions to all facilities. This is for later
eventually updating the anticipated times of the events in the heaps
(that takes time $O\left(n_{f}\log n_{c}n_{f}\right)$) following
}each Action \textcolor{black}{1-c. This} action dominates\textcolor{black}{{}
the overall runtime complexity. Hence the total time is $O\left(n_{c}^{2}n_{f}^{2}\right)$.}
\end{proof}
The algorithm computes the same solution as the PD algorithm, so we
have the following theorem.
\begin{theorem}
Algorithm APD is 1.61-, (1.11, 1.78)- and (1,2)-approximation in time
$O\left(n^{4}\right)$ for the uniform $FTRA$.
\end{theorem}
In order to further achieve the factor of 1.52 in strongly polynomial
time that matches the best result \cite{Swamy08FTFL2.076} for the
uniform $FTFL$, it is necessary to apply the cost scaling and greedy
augmentation (GA) techniques \cite{Swamy08FTFL2.076,Guha03FTFL2.41}
for $FTFL$ to $FTRA$. However, like in \cite{kewen2011cocoon,yan2011approximation},
the difficulty encountered is the application of greedy augmentation
(GA) in polynomial time, since the naive way of treating an $FTRA/FTRA_{\infty}$
instance as an equivalent $FTFL$ instance and then directly applying
GA after cost scaling will result in weakly polynomial or pseudo-polynomial
time algorithms, depending on whether using the instance shrinking
technique in the previous section or not. 

\begin{algorithm}[H]
\caption{AGA: Acceleration of Greedy Augmentation}

\textbf{\textcolor{black}{Input}}\textcolor{black}{: }$\mathcal{F},\,\mathcal{C},\,\boldsymbol{f},\,\boldsymbol{c},\,\boldsymbol{r},\,\boldsymbol{R}$\textcolor{black}{,}\textbf{\textcolor{black}{{}
}}$\left(\boldsymbol{x},\,\boldsymbol{y}\right)$. \textbf{\textcolor{black}{Output:
}}$\left(\boldsymbol{x},\,\boldsymbol{y}\right)$.

\textbf{\textcolor{black}{Initialization}}\textcolor{black}{: }

\textbf{for $j\in\mathcal{C}$ }//optimize the total connection cost
first

\qquad{}\textbf{for $i\in\mathcal{F}$ }and $y_{i}>0$, in the increasing
order of distances w.r.t $j$

\qquad{}\qquad{}$x_{ij}\leftarrow\min\left(r_{j},\, y_{i}\right)$

\qquad{}\qquad{}$r_{j}\leftarrow r_{j}-x_{ij}$

\textbf{set} residual vector $\boldsymbol{\bar{y}}\leftarrow\boldsymbol{R}-\boldsymbol{y}$
//for detecting the case $y_{i}$ reaches $R_{i}$

\textbf{set} $CC\leftarrow\sum_{i\in\mathcal{F}}\sum_{j\in\mathcal{C}}c_{ij}x_{ij}$
as the current total connection cost

\textbf{invoke} calculateGain

\medskip{}

\textbf{while} $\max_{i\in\mathcal{F}}gain\left(i\right)>0$: //if
$ $$gain\left(i\right)>0$, then $\bar{y_{i}}>0$ from the calculateGain
function

\qquad{}\textbf{pick} $i^{*}=\arg\max_{i\in\mathcal{F}}\frac{gain\left(i\right)}{f_{i}}$

\qquad{}$\mathcal{S}_{i}\leftarrow\left\{ j\in\mathcal{\mathcal{\mathcal{C}}}\,|\,\max_{i'\in\mathcal{F}\&\&x_{i'j}>0}c_{i'j}-c_{i^{*}j}>0\right\} $

\qquad{}$\forall j\in\mathcal{S}_{i}:\, i_{j}^{*}\leftarrow\arg\max_{i'\in\mathcal{F}\&\&x_{i'j}>0}c_{i'j}$ 

\qquad{}$NS\leftarrow\min_{j\in\mathcal{S}_{i}}x_{i_{j}^{*}j}$,
$ToC\leftarrow\min\left(NS,\,\bar{y_{i}}\right)$

\qquad{}\textbf{\textcolor{black}{set}} $y_{i^{*}}\leftarrow y_{i^{*}}+ToC$

\qquad{}$\Delta\leftarrow0$ //$\Delta$ stores the total decrease
in the connection cost after all switches

\qquad{}\textbf{for $j\in\mathcal{S}_{i}$}

\qquad{}\qquad{}$\Delta\leftarrow\Delta+ToC\cdot\left(\max_{i'\in\mathcal{F}\&\&x_{i'j}>0}c_{i'j}-c_{i^{*}j}\right)$

\qquad{}\qquad{}\textbf{\textcolor{black}{set}} $x_{i_{j}^{*}j}\leftarrow x_{i_{j}^{*}j}-ToC$
and $x_{i^{*}j}\leftarrow x_{i^{*}j}+ToC$

\qquad{}\textbf{\textcolor{black}{set}} $CC\leftarrow CC-\Delta$

\qquad{}\textbf{update }$\boldsymbol{\bar{y}}$

\qquad{}\textbf{invoke} calculateGain

\medskip{}

\textbf{function} calculateGain

\qquad{}\textbf{for $i\in\mathcal{F}$}

\qquad{}\qquad{}$C_{i}\leftarrow CC$ //for each $i$, $C_{i}$
stores the total connection cost after connections 

\qquad{}\qquad{}are switched to $i$

\qquad{}\qquad{}$gain\left(i\right)\leftarrow0$

\qquad{}\qquad{}\textbf{if }$\bar{y_{i}}>0$

\qquad{}\qquad{}\qquad{}\textbf{for $j\in\mathcal{C}$}

\qquad{}\qquad{}\qquad{}\qquad{}\textbf{if }$\max_{i'\in\mathcal{F}\&\&x_{i'j}>0}c_{i'j}>c_{ij}$

\qquad{}\qquad{}\qquad{}\qquad{}\qquad{}$C_{i}\leftarrow C_{i}-\max_{i'\in\mathcal{F}\&\&x_{i'j}>0}c_{i'j}+c_{ij}$

\qquad{}\qquad{}\qquad{}$gain\left(i\right)\leftarrow CC-C_{i}-f_{i}$
\end{algorithm}

Nevertheless, if GA is applied with another similar acceleration heuristic,
it changes to the algorithm AGA\textit{ }(Algorithm 4) which runs
in strongly polynomial time. Before describing AGA, we take a brief
look at GA in \cite{Guha03FTFL2.41} for $FTFL$. It defines $ $$gain\left(i\right)$
of a facility $i$ to be the decrease in total cost (decrease in total
connection cost minus increase in facility cost of $i$) of the solution
after adding a facility $i$ to open and connecting clients to their
closest facilities. Note that once a set of open facilities are fi{}xed,
the total connection cost can be easily computed since every client
simply chooses these facilities in increasing order of distance. GA
then iteratively picks the facility with the largest gain ratio $\frac{gain\left(i\right)}{f_{i}}$
to open until there is no facility $i$ with $gain\left(i\right)>0$
left. On the other hand, AGA computes $gain\left(i\right)$ in the
same way as GA. The difference is in $FTRA$ there are \textcolor{black}{$\sum_{i\in\mathcal{F}}R_{i}$
facilities in total, therefore it is slow to consider one facility
at a time (in each iteration of AGA). Fortunately, there is also an
acceleration heuristic: because all facilities at a site $i$ has
$gain\left(i\right)$, once a facility at site $i_{m}$ with} $\max_{i}\frac{gain\left(i\right)}{f_{i}}$\textcolor{black}{{}
is selected to open, additional facilities at $i_{m}$ may also open
at the same time until either (1) this maximum} \textcolor{black}{$gain\left(i_{m}\right)$}
reduces due to insufficient decrease in the total connection cost;
or (2) $y_{i}$ reaches $R_{i}$. Moreover, (1) \textcolor{black}{happens
once a client has appeared to switch all of its most expensive connections
to $i_{m}$, which is similar to the switching case in the previous
algorithm APD.}

Formally in the AGA algorithm, $CC$ denotes the current total connection
cost and $C_{i}$ the connection cost after $i$ is opened and client
connections are switched. The calculateGain function computes $gain\left(i\right)$
and the while loop implements GA with the described heuristic. In
each loop iteration, for updating $CC$, $\Delta$ stores the total
decrease in the connection cost after client switching. Following
the heuristic, $ToC$ and $NS$ are defined similarly as in the APD
algorithm. Note that in the initialization phase of AGA, the total
connection cost is optimized first so that every client connects to
its closest facilities. This is to ensure that in every iteration
only the client connections with the largest costs need to be considered
in computing the best possible connection cost $C_{i}$.
\begin{lemma}
Algorithm AGA runs in $O\left(n^{4}\right)$ for $FTRA$.\label{lem:time-GA}\end{lemma}
\begin{proof}
Each iteration of the while loop runs in $O\left(n_{c}n_{f}\right)$
due to the calculateGain function. Now, we bound the the total number
of iterations. Similar to the acceleration heuristic analysis of the
algorithm APD (c.f. Lemma \ref{lem:suh}), in AGA once a site $i_{m}$
with the maximum gain is chosen, AGA opens the facilities at $i_{m}$
until either $R_{i_{m}}$ is reached, or \textcolor{black}{a client
has appeared to switch all of its most expensive connections, causing
reduced maximum gain. Further, there are at most $n_{f}$ chances
to reach }$R_{i_{m}}$\textcolor{black}{{} and $n_{c}n_{f}$ possible
sets of connections for all clients. Since clients also only switch
to cheaper connections, there are maximum $\left(n_{f}+n_{c}n_{f}\right)$
iterations.} The total time is therefore $O\left(n_{c}^{2}n_{f}^{2}\right)$.
\end{proof}
Now the important observation/trick for the analysis is that applying
AGA to an $FTRA/FTRA_{\infty}$ instance (with solution) obtains essentially
the \textit{same solution} (also the same cost) as treating this instance
as an equivalent $FTFL$ instance (by naively splitting sites) and
then directly applying GA. The difference is, with the acceleration
heuristic, AGA is able to arrive at this solution faster, in strongly
polynomial time. The observation then implies that AGA alone improves
the 3.16-approximation result of \cite{yan2011approximation} for
the general $FTRA_{\infty}$ to 2.408 in polynomial time using the
GA results \cite{Guha03FTFL2.41} for $FTFL$. Similarly, for the
uniform $FTRA$, AGA combined with cost scaling further improves on
the (1.11, 1.78)-approximation algorithm APD according to the results
of \cite{Swamy08FTFL2.076} for the uniform $FTFL$.
\begin{theorem}
The uniform $FTRA$ can be approximated with a factor of 1.52 in time
$O\left(n^{4}\right)$. \label{thm:1.52uniform}
\end{theorem}

\section{The Uniform $k$-$FTRA$}

Lastly, we consider the Constrained Fault-Tolerant $k$-Resource Allocation
($k$-$FTRA$) problem and show its uniform case achieves an approximation
ratio of 4. In this important variant of $FTRA$, there is an additional
constraint that at most $k$ facilities ($\max_{j\in\mathcal{C}}r_{j}\leq k$
and $k\leq\sum_{i\in\mathcal{F}}R_{i}$) across all sites can be opened
as resources. This problem has the following formulation.

\textit{
\begin{equation}
\begin{array}{llc}
\mathrm{minimize} & \sum_{i\in\mathcal{F}}f_{i}y_{i}+\sum_{i\in\mathcal{F}}\sum_{j\in\mathcal{C}}c_{ij}x_{ij}\\
\mathrm{subject\, to} & \forall j\in\mathcal{C}:\,\sum_{i\in\mathcal{F}}x_{ij}\ge r_{j}\\
 & \forall i\in\mathcal{F},j\in\mathcal{C}:\, y_{i}-x_{ij}\geq0\\
 & \sum_{i\in\mathcal{F}}y_{i}\leq k\\
 & \forall i\in\mathcal{F}:\, y_{i}\leq R_{i}\\
 & \forall i\in\mathcal{F},j\in\mathcal{C}:\, x_{ij},\, y_{i}\in\mathbb{Z}^{+}
\end{array}\label{eq:kftra-ip}
\end{equation}
}

Its LP-relaxation (primal LP) and dual LP are:

\textit{
\begin{equation}
\begin{array}{llc}
\mathrm{minimize} & \sum_{i\in\mathcal{F}}f_{i}y_{i}+\sum_{i\in\mathcal{F}}\sum_{j\in\mathcal{C}}c_{ij}x_{ij}\\
\mathrm{subject\, to} & \forall j\in\mathcal{C}:\,\sum_{i\in\mathcal{F}}x_{ij}\ge r_{j}\\
 & \forall i\in\mathcal{F},j\in\mathcal{C}:\, y_{i}-x_{ij}\geq0\\
 & \sum_{i\in\mathcal{F}}y_{i}\leq k\\
 & \forall i\in\mathcal{F}:\, y_{i}\leq R_{i}\\
 & \forall i\in\mathcal{F},j\in\mathcal{C}:\, x_{ij},\, y_{i}\geq0
\end{array}\label{eq:kftra-lp}
\end{equation}
}

\begin{equation}
\begin{array}{llc}
\textrm{maximize} & \sum_{j\in\mathcal{C}}r_{j}\alpha_{j}-\sum_{i\in\mathcal{F}}R_{i}z_{i}-k\theta\\
\mathrm{subject\, to} & \forall i\in\mathcal{F}:\,\sum_{j\in\mathcal{C}}\beta_{ij}\leq f_{i}+z_{i}+\theta\\
 & \forall i\in\mathcal{F},j\in\mathcal{C}:\,\alpha_{j}-\beta_{ij}\leq c_{ij}\\
 & \forall i\in\mathcal{F},j\in\mathcal{C}:\,\alpha_{j},\,\beta_{ij},\, z_{i},\,\theta\geq0
\end{array}\label{eq:kftra-dual}
\end{equation}

It is clear that $k$-$FTRA$ generalizes the well studied $k$-$UFL$
\cite{jain01approximation,Jain03dualfitting} and $k$-$FTFL$ \cite{Swamy08FTFL2.076}
problems. In the following, besides adapting the algorithms and analyses
therein, we also develop a greedy pairing (GP) procedure which in
polynomial time constructs paired and unpaired sets of facilities
from sites for randomly opening them afterwards.

\textbf{Algorithm Description. }The algorithm PK (Algorithm 5) consists
of three sequential procedures: Binary Search (BS), Greedy Pairing
(GP) and Randomized Rounding (RR). BS utilizes the previous $\left(1,\,2\right)$-approximation
algorithm APD (Algorithm 3) for $FTRA$ with the \textit{modified}
input facility cost $2\left(f_{i}+\theta\right)$, i.e. the cost is
increased by $\theta$ first and then scaled by $2$. As we will see
later in the analysis, this modification is necessary for two reasons:
1) the Lagrangian relaxation of $k$-$FTRA$ is $FTRA$; 2) the scaling
of the facility cost enables us to build on the approximation ratio
$\left(1,\,2\right)$ of $FTRA$ for getting the ratio of $k$-$FTRA$.
For simplicity, let APD$\left(\theta,\,\lambda\right)$ denote the
parameterized APD algorithm with the input facility cost perturbing
factor $\theta$ and scaling factor $\lambda$, so APD$\left(0,\,1\right)$
produces the same solution as APD. From LP \eqref{eq:ftra-ip} and
\eqref{eq:kftra-ip}, it is clear that APD produces an almost feasible
integral solution to $k$-$FTRA$ except that it has to guarantee
at most $k$ facilities to open ($\sum_{i\in\mathcal{F}}y_{i}\leq k$)
from all sites. This guarantee might not be even possible, but fortunately
we can use APD$\left(\theta,\,\lambda\right)$ to get two solutions
$\left(\boldsymbol{x_{s}},\,\boldsymbol{y_{s}}\right)$ and $\left(\boldsymbol{x_{l}},\,\boldsymbol{y_{l}}\right)$
$ $with the small one having $\sum_{i\in\mathcal{F}}y_{s,i}=k_{s}<k$
and the large one $\sum_{i\in\mathcal{F}}y_{l,i}=k_{l}>k$ facilities
open. A convex combination of these two solutions is able to give
a feasible \textit{fractional} solution $\left(\boldsymbol{x'},\,\boldsymbol{y'}\right)$
to $k$-$FTRA$ instead, i.e. $\left(\boldsymbol{x'},\,\boldsymbol{y'}\right)=a\left(\boldsymbol{x_{s}},\,\boldsymbol{y_{s}}\right)+b\left(\boldsymbol{x_{l}},\,\boldsymbol{y_{l}}\right)$
with $a+b=1$ and $ak_{s}+bk_{l}=k$. The solutions can be obtained
by binary searching two values ($\theta_{1}$ and $\theta_{2}$) of
$\theta$ over the interval $\left[0,\,\frac{n_{c}c_{max}}{\lambda}\right]$
where $c_{max}=\max_{i\in\mathcal{F},\, j\in\mathcal{C}}c_{ij}$ and
invoking APD$\left(\theta_{1},\,\lambda\right)$ and APD$\left(\theta_{2},\,\lambda\right)$.
This specific interval is chosen because as the value of $\theta$
increases, the number of open facilities from APD$\left(\theta,\,\lambda\right)$
will decrease. So if $\theta=\frac{n_{c}c_{max}}{\lambda}$, the algorithm
will only open the minimum number of $\max_{j\in\mathcal{C}}r_{j}$
facilities.%
\footnote{We noticed that the binary search interval $\left[0,\, nrc_{max}\right]$
(c.f. the third paragraph of Section 7 of \cite{Swamy08FTFL2.076})
for $k$-$FTFL$ can be reduced to $\left[0,\,\frac{n_{c}c_{max}}{2}\right]$,
because once the minimum number of $\max_{j\in\mathcal{C}}r_{j}$
facilities are opened and all facility costs are at least $n_{c}c_{max}$,
from the primal-dual algorithm, all clients are already fully-connected.%
} Moreover, as shown later, if $\theta_{1}$ and $\theta_{2}$ become
sufficiently close ($\epsilon=\frac{c_{\min}}{8N^{2}}$ where $c_{\min}$
is the smallest positive connection cost and $N=\sum_{i\in\mathcal{F}}R_{i}$
\footnote{The algorithm analysis also holds if $N=\sum_{j\in\mathcal{C}}r_{j}$.%
}) in BS, the approximation ratio of APD is almost preserved for building
a ratio for $k$-$FTRA$.

\begin{algorithm}
\caption{PK: Procedures for $k$-$FTRA$}

\textbf{\textcolor{black}{Input}}\textcolor{black}{: }A $k$-$FTRA$
instance $\left\langle \mathcal{F},\,\mathcal{C},\, f_{i},\, c_{ij},\, r_{j},\, R_{i},\, k\right\rangle $\textcolor{black}{.}\textbf{\textcolor{black}{{}
Output: }}$\left(\boldsymbol{x},\,\boldsymbol{y}\right)$

\textbf{\textcolor{black}{Initialization}}\textcolor{black}{:} $\boldsymbol{x}\leftarrow\boldsymbol{0},\,\boldsymbol{y}\leftarrow\boldsymbol{0}$

\textbf{\textcolor{black}{Procedure 1}}\textcolor{black}{: }Binary
Search (BS)

$\theta_{1}=0,\,\theta_{2}=\frac{n_{c}\max_{i\in\mathcal{F},\, j\in\mathcal{C}}c_{ij}}{2}$

\textbf{\textcolor{black}{while}}\textcolor{black}{{} $\theta_{2}-\theta_{1}>\epsilon$
}\textbf{\textcolor{black}{do}}\textcolor{black}{:}

\qquad{}$mid=\frac{\theta_{2}-\theta_{1}}{2}$

\qquad{}\textbf{invoke} APD with $\left\langle \mathcal{F},\,\mathcal{C},\,2\left(f_{i}+mid\right),\, c_{ij},\, r_{j},\, R_{i}\right\rangle $
and output $\left(\boldsymbol{x_{mid}},\,\boldsymbol{y_{mid}}\right)$

\qquad{}$k_{mid}=\sum_{i\in\mathcal{F}}y_{mid,i}$

\qquad{}\textbf{if }$k_{mid}<k$

\qquad{}\qquad{}\textbf{set $\theta_{2}=mid$}

\qquad{}\textbf{else if }$k_{mid}>k$

\qquad{}\qquad{}\textbf{set $\theta_{1}=mid$}

\qquad{}\textbf{else}

\qquad{}\qquad{}\textbf{return $mid$ }//if here reached, all procedures
afterwards can be safely ignored

\textbf{invoke} APD with $\left\langle \mathcal{F},\,\mathcal{C},\,2\left(f_{i}+\theta_{1}\right),\, c_{ij},\, r_{j},\, R_{i}\right\rangle $
and output $\left(\boldsymbol{x_{l}},\,\boldsymbol{y_{l}}\right)$

\textbf{invoke} APD with $\left\langle \mathcal{F},\,\mathcal{C},\,2\left(f_{i}+\theta_{2}\right),\, c_{ij},\, r_{j},\, R_{i}\right\rangle $
and output $\left(\boldsymbol{x_{s}},\,\boldsymbol{y_{s}}\right)$

$k_{l}=\sum_{i\in\mathcal{F}}y_{l,i}>k$ and $k_{s}=\sum_{i\in\mathcal{F}}y_{s,i}<k$

\medskip{}

\textbf{\textcolor{black}{Procedure 2}}\textcolor{black}{: Greedy
Pairing (GP)}

//vectors representing numbers of constructed paired and unpaired
facilities in $\boldsymbol{y_{l}}$

\textbf{set} $\boldsymbol{y_{p}}\leftarrow\boldsymbol{0}$, $\boldsymbol{\bar{y_{p}}}\leftarrow\boldsymbol{0}$ 

\textbf{for $i\in\mathcal{F}$}

\qquad{}if $y_{s,i}>0$ and $y_{l,i}>0$

\qquad{}\qquad{}$y_{p,i}\leftarrow y_{p,i}+\min\left(y_{s,i},\, y_{l,i}\right)$

\qquad{}\qquad{}//updating vectors and store them in $\boldsymbol{\hat{y_{s}}}$
and $\boldsymbol{\hat{y_{l}}}$ for the next pairing steps

\qquad{}\qquad{}$\hat{y_{s,i}}\leftarrow y_{s,i}-\min\left(y_{s,i},\, y_{l,i}\right)$

\qquad{}\qquad{}$\hat{y_{l,i}}\leftarrow y_{l,i}-\min\left(y_{s,i},\, y_{l,i}\right)$

\textbf{for $i\in\mathcal{F}$ }in arbitrary order

\qquad{}\textbf{if} $\hat{y_{s,i}}>0$

\qquad{}\qquad{}\textbf{for $i'\in\mathcal{F}\backslash i$ }in
the order of closest to $i$

\qquad{}\qquad{}\qquad{}\textbf{if} $\hat{y_{l,i'}}>0$

\qquad{}\qquad{}\qquad{}\qquad{}$y_{p,i'}\leftarrow y_{p,i'}+\min\left(\hat{y_{s,i}},\,\hat{y_{l,i'}}\right)$

\qquad{}\qquad{}\qquad{}\qquad{}$\hat{y_{s,i}}\leftarrow\hat{y_{s,i}}-\min\left(\hat{y_{s,i}},\,\hat{y_{l,i'}}\right)$

\qquad{}\qquad{}\qquad{}\qquad{}$\hat{y_{l,i'}}\leftarrow\hat{y_{l,i'}}-\min\left(\hat{y_{s,i}},\,\hat{y_{l,i'}}\right)$

$\boldsymbol{\bar{y_{p}}}\leftarrow\boldsymbol{y_{l}}-\boldsymbol{y_{p}}$
//at this time $\sum_{i\in\mathcal{F}}y_{p,i}=k_{s}$ and $\sum_{i\in\mathcal{F}}\bar{y_{p,i}}=k_{l}-k_{s}$

\medskip{}

\textbf{\textcolor{black}{Procedure }}\textcolor{black}{3: Randomized
Rounding (RR)}

\textbf{choose }probabilities $a=\frac{k_{l}-k}{k_{l}-k_{s}}$ and
$b=\frac{k-k_{s}}{k_{l}-k_{s}}$ so $ak_{s}+bk_{l}=k$ and $a+b=1$

\textbf{set} $\boldsymbol{y}\leftarrow\boldsymbol{y_{s}}$ with probability
$a$ and $\boldsymbol{y}\leftarrow\boldsymbol{y_{p}}$ with probability
$b=1-a$ //disjoint cases both open $k_{s}$ facilities

\textbf{select }a random subset of $k-k_{s}$ facilities to open from
$\boldsymbol{\bar{y_{p}}}$ and add these to $\boldsymbol{y}$ //at
this time $\sum_{i\in\mathcal{F}}y_{i}=k$ and each facility in $\boldsymbol{y_{l}}$
is opened with probability $b$

//connects each client $j$ to its closest $r_{j}$ opened facilities

\textbf{for $j\in\mathcal{C}$}

\qquad{}\textbf{for $i\in\mathcal{F}$ }in the order of closest to
$j$

\qquad{}\qquad{}$x_{ij}\leftarrow\min\left(r_{j},\, y_{i}\right)$

\qquad{}\qquad{}$r_{j}\leftarrow r_{j}-x_{ij}$
\end{algorithm}

However, for a feasible \textit{integral} solution $\left(\boldsymbol{x},\,\boldsymbol{y}\right)$
with $k$ open facilities, the algorithm instead relies on our\textit{
efficient} GP and RR procedures. These procedures extend the matching
and rounding procedures (cf. the paragraph before Lemma 7.1 in \cite{Swamy08FTFL2.076})
for $k$-$FTFL$ respectively. In particular, based on the solution
vectors $\boldsymbol{y_{s}}$ and $\boldsymbol{y_{l}}$ obtained from
BS, GP splits the vector $\boldsymbol{y_{l}}$ into $\boldsymbol{y_{p}}$
and $\boldsymbol{\bar{y_{p}}}$ s.t. $\boldsymbol{y_{l}}=\boldsymbol{y_{p}}+\boldsymbol{\bar{y_{p}}}$
and $\sum_{i\in\mathcal{F}}y_{s,i}=\sum_{i\in\mathcal{F}}y_{p,i}=k_{s}$.
Note that each of these integral vectors represents the facility opening
amounts of all sites. To be precise, GP greedily constructs the paired
($\boldsymbol{y_{p}}$) and unpaired facilities ($\boldsymbol{\bar{y_{p}}}$)
from $\boldsymbol{y_{l}}$ against the small solution $\boldsymbol{y_{s}}$.
It first pairs the facilities of the corresponding sites in $\boldsymbol{y_{s}}$
and $\boldsymbol{y_{l}}$ (both sites with open facilities) and records
the pairing result in $\boldsymbol{y_{p}}$. Next, for each left unpaired
site $ $$i$ in $\boldsymbol{\hat{y_{s}}}$ in arbitrary order, GP
exhaustively pairs the facilities at $i$ with the facilities of the
unpaired sites in $\boldsymbol{\hat{y_{l}}}$ in the order of closest
to $i$. In this pairing step, $\boldsymbol{y_{p}}$ is updated accordingly.
At the end, $\boldsymbol{\bar{y_{p}}}$ is simply set to be $\boldsymbol{y_{l}}-\boldsymbol{y_{p}}$.
To be more precise, we consider a simple example with $\boldsymbol{y_{s}}=\left[3,\,2,\,0,\,2\right]$
and $\boldsymbol{y_{l}}=\left[2,\,0,\,5,\,3\right]$ from BS before
running GP. After the first pairing step, $\boldsymbol{y_{p}}=\left[2,\,0,\,0,\,2\right]$,
$\hat{\boldsymbol{y_{s}}}=\left[1,\,2,\,0,\,0\right]$ and $\hat{\boldsymbol{y_{l}}}=\left[0,\,0,\,5,\,1\right]$.
Now for simplicity, we assume that the distance between sites $i$
and $j$ is $\left|i-j\right|$ where we follow the ascending order
of indices $j$'s in resolving the ties of the closest distances.
Therefore, after the second step, $\boldsymbol{y_{p}}=\left[2,\,0,\,3,\,2\right]$,
$\hat{\boldsymbol{y_{s}}}=\left[0,\,0,\,0,\,0\right]$ and $\hat{\boldsymbol{y_{l}}}=\bar{\boldsymbol{y_{p}}}=\left[0,\,0,\,2,\,1\right]$,
since both the unpaired $\hat{y_{s,1}}$ and $\hat{y_{s,2}}$ (1-based
index) are paired to the closest unpaired $\hat{y_{l,3}}$.

From the $\boldsymbol{y_{s}}$, $\boldsymbol{y_{p}}$ and $\boldsymbol{\bar{y_{p}}}$
obtained, the last procedure RR then randomly opens $k$ facilities
in a way that the expected facility opening cost of $\boldsymbol{y}$
is the \textit{same} as the facility opening cost of the convex combination
solution $\boldsymbol{y'}$. In addition, RR connects each client
$j$ to its closest $r_{j}$ open facilities in $\boldsymbol{y}$,
ensuring the expected connection cost of $\boldsymbol{x}$ is\textit{
bounded} by the connection cost of $\boldsymbol{x'}$.

\textbf{Algorithm Analysis. }The basic idea of the analysis is to
first bound $cost\left(\boldsymbol{x'},\,\boldsymbol{y'}\right)$
by $cost\left(\boldsymbol{\alpha'},\,\boldsymbol{\beta'},\,\boldsymbol{z'},\theta'\right)$
where $\left(\boldsymbol{\alpha'},\,\boldsymbol{\beta'},\,\boldsymbol{z'},\theta'\right)$
is a constructed feasible dual solution to LP \eqref{eq:kftra-dual}.
Then we bound the expected total cost $cost\left(\boldsymbol{x},\,\boldsymbol{y}\right)$
with $cost\left(\boldsymbol{x'},\,\boldsymbol{y'}\right)$ to further
establish the approximation ratio $\rho$ s.t. $cost\left(\boldsymbol{x},\,\boldsymbol{y}\right)\leq\rho cost\left(\boldsymbol{\alpha'},\,\boldsymbol{\beta'},\,\boldsymbol{z'},\theta'\right)$.
Finally, by the weak duality theorem, $cost\left(\boldsymbol{x},\,\boldsymbol{y}\right)\leq\rho cost\left(\boldsymbol{x_{k}^{*}},\,\boldsymbol{y_{k}^{*}}\right)$
where $\left(\boldsymbol{x_{k}^{*}},\,\boldsymbol{y_{k}^{*}}\right)$
is the optimal fractional solution to $k$-$FTRA$ (displayed as LP
\eqref{eq:kftra-lp}).

In the first step, we focus on analyzing the BS procedure to bound
$cost\left(\boldsymbol{x'},\,\boldsymbol{y'}\right)$ by $cost\left(\boldsymbol{\alpha'},\,\boldsymbol{\beta'},\,\boldsymbol{z'},\theta'\right)$.
Suppose APD$\left(\theta,\,2\right)$ produces the primal solution
$\left(\boldsymbol{\tilde{x}},\,\boldsymbol{\tilde{y}}\right)$ with
$\tilde{k}$ open facilities. We let the cost of $\left(\boldsymbol{\tilde{x}},\,\boldsymbol{\tilde{y}}\right)$
w.r.t. the \textit{original} input instance be $cost\left(\boldsymbol{\tilde{x}},\,\boldsymbol{\tilde{y}}\right)=\tilde{F}+\tilde{C}$,
where in the separate costs $\left(\tilde{F},\,\tilde{C}\right)$,
$\tilde{F}=\sum_{i\in\mathcal{F}}f_{i}\tilde{y_{i}}$ is the total
facility cost and $\tilde{C}=\sum_{i\in\mathcal{F}}\sum_{j\in\mathcal{C}}c_{ij}\tilde{x_{ij}}$
is the connection cost. Similarly, w.r.t. the modified instance, the
cost is $cost'\left(\boldsymbol{\tilde{x}},\,\boldsymbol{\tilde{y}}\right)=2\left(\tilde{F}+\tilde{k}\theta\right)+\tilde{C}$.
From the analysis (cf. the paragraph before Theorem \ref{thm: multifactor})
of the factor revealing program of the PD algorithm, for APD$\left(\theta,\,2\right)$,
we get $\forall i\in\mathcal{F}:\,\sum_{j\in\mathcal{C}'}\left(u_{ij}-2c_{ij}\right)\leq2\left(f_{i}+\theta\right)$
where $\mathcal{C}'\subseteq\mathcal{C}$, i.e.,

\begin{equation}
\sum_{j\in\mathcal{C}}\tilde{\alpha_{j}}-2\sum_{j\in\mathcal{C}}c_{ij}-\tilde{z_{i}}-2\left(f_{i}+\theta\right)\leq0
\end{equation}

where $\left(\boldsymbol{\tilde{\alpha}},\,\boldsymbol{\tilde{\beta}},\,\boldsymbol{\tilde{z}}\right)$
is the corresponding constructed dual values of $\left(\boldsymbol{\tilde{x}},\,\boldsymbol{\tilde{y}}\right)$
from the PD algorithm. Further, from Lemma \ref{lem: pdb}, we have
a bound for $cost'\left(\boldsymbol{\tilde{x}},\,\boldsymbol{\tilde{y}}\right)$,
i.e.,

$ $
\begin{equation}
2\left(\tilde{F}+\tilde{k}\theta\right)+\tilde{C}\leq\sum_{j\in\mathcal{C}}r_{j}\tilde{\alpha_{j}}-\sum_{i\in\mathcal{F}}R_{i}\tilde{z_{i}}.
\end{equation}
Note that the dual solution $\left(\boldsymbol{\tilde{\alpha}},\,\boldsymbol{\tilde{\beta}},\,\boldsymbol{\tilde{z}}\right)$
is used only in the analysis. Also, because APD only speeds up PD
by combining its events, we can use the dual solution produced from
PD for analyzing APD. If we set $\left(\boldsymbol{\tilde{\alpha'}},\,\boldsymbol{\tilde{z'}},\,\tilde{\theta'}\right)=\left(\frac{\boldsymbol{\tilde{\alpha}}}{2},\,\frac{\boldsymbol{\tilde{z}}}{2},\,\theta\right)$
and $\forall i\in\mathcal{F},\, j\in\mathcal{C}:\,\tilde{\beta_{ij}'}=\tilde{\alpha_{j}'}-c_{ij}$,
the inequality (7) then becomes $\sum_{j\in\mathcal{C}}\tilde{\beta_{ij}'}\leq f_{i}+\tilde{z_{i}'}+\tilde{\theta'}$,
implying $\left(\boldsymbol{\tilde{\alpha'}},\,\boldsymbol{\tilde{\beta'}},\,\boldsymbol{\tilde{z'}},\tilde{\theta'}\right)$
is a feasible dual solution to LP \eqref{eq:kftra-dual}. Furthermore,
(8) becomes

\begin{equation}
2\tilde{F}+\tilde{C}\leq2\left(\sum_{j\in\mathcal{C}}r_{j}\tilde{\alpha_{j}'}-\sum_{i\in\mathcal{F}}R_{i}\tilde{z_{i}'}-\tilde{k}\tilde{\theta'}\right).
\end{equation}
The analysis here reveals the Lagrangian relation between $k$-$FTRA$
and $FTRA$ from the dual perspective, whereas the Lagrangian relaxation
framework (cf. Section 3.6 of \cite{jain01approximation}) starts
from the primal. Therefore, if $\tilde{k}=k$, $\left(\boldsymbol{\tilde{x}},\,\boldsymbol{\tilde{y}}\right)$
is 2-approximation from the inequality (9), the bound $cost\left(\boldsymbol{\tilde{x}},\,\boldsymbol{\tilde{y}}\right)<2\tilde{F}+\tilde{C}$
and the feasibilities of $\left(\boldsymbol{\tilde{x}},\,\boldsymbol{\tilde{y}}\right)$
and $\left(\boldsymbol{\tilde{\alpha'}},\,\boldsymbol{\tilde{\beta'}},\,\boldsymbol{\tilde{z'}},\tilde{\theta'}\right)$.
However, as mentioned before, we may never encounter the situation
$\tilde{k}=k$. Instead, the BS procedure finds $ $$\theta_{1}$
and $\theta_{2}$ until $\theta_{2}-\theta_{1}\leq\epsilon=\frac{c_{min}}{8N^{2}}$.
It then runs APD$\left(\theta_{1},\,2\right)$ to obtain the solution
$\left(\boldsymbol{x_{l}},\,\boldsymbol{y_{l}}\right)$ with $k_{l}>k$
and the cost $\left(F_{l},\, C_{l}\right)$ w.r.t. the original instance;
and APD$\left(\theta_{2},\,2\right)$ to get the solution $\left(\boldsymbol{x_{s}},\,\boldsymbol{y_{s}}\right)$
with $k_{s}<k$ and $\left(F_{s},\, C_{s}\right)$. Hence, from (9)
we have

\begin{equation}
2F_{l}+C_{l}\leq2\left(\sum_{j\in\mathcal{C}}r_{j}\alpha_{l,j}'-\sum_{i\in\mathcal{F}}R_{i}z_{l,i}'-k_{l}\theta_{1}\right)
\end{equation}

and

\begin{equation}
2F_{s}+C_{s}\leq2\left(\sum_{j\in\mathcal{C}}r_{j}\alpha_{s,j}'-\sum_{i\in\mathcal{F}}R_{i}z_{s,i}'-k_{s}\theta_{2}\right)
\end{equation}

where $\left(\boldsymbol{\alpha_{l}'},\,\boldsymbol{\beta_{l}'},\,\boldsymbol{z_{l}'}\right)$
and $\left(\boldsymbol{\alpha_{s}'},\,\boldsymbol{\beta_{s}'},\,\boldsymbol{z_{s}'}\right)$
are constructed as $\left(\boldsymbol{\tilde{\alpha'}},\,\boldsymbol{\tilde{\beta'}},\,\boldsymbol{\tilde{z'}}\right)$
to be feasible duals.

Now we are ready to bound $cost\left(\boldsymbol{x'},\,\boldsymbol{y'}\right)$
by $cost\left(\boldsymbol{\alpha'},\,\boldsymbol{\beta'},\,\boldsymbol{z'},\theta'\right)$.
The proof of the following lemma builds on the idea of Lemma 9 in
\cite{jain01approximation} for $k$-median.
\begin{lemma}
\textup{$cost\left(\boldsymbol{x'},\,\boldsymbol{y'}\right)<\left(2+\frac{1}{N}\right)F'+C'\leq\left(2+\frac{1}{N}\right)cost\left(\boldsymbol{\alpha'},\,\boldsymbol{\beta'},\,\boldsymbol{z'},\theta'\right)$,}
where $N=\sum_{i\in\mathcal{F}}R_{i}$,\textup{ $\left(\boldsymbol{x'},\,\boldsymbol{y'}\right)=a\left(\boldsymbol{x_{s}},\,\boldsymbol{y_{s}}\right)+b\left(\boldsymbol{x_{l}},\,\boldsymbol{y_{l}}\right)$,
$a+b=1,\, k=ak_{s}+bk_{l},$ $F'=\sum_{i\in\mathcal{F}}f_{i}y_{i}'=aF_{s}+bF_{l}$,
$C'=\sum_{i\in\mathcal{F}}\sum_{j\in\mathcal{C}}c_{ij}x_{ij}'=aC_{s}+bC_{l}$},
$\boldsymbol{\alpha'}=a\boldsymbol{\alpha_{s}'}+b\boldsymbol{\alpha_{l}'}$,
$\boldsymbol{\beta'}=a\boldsymbol{\beta_{s}'}+b\boldsymbol{\beta_{l}'}$,
$\boldsymbol{z'}=a\boldsymbol{z_{s}'}+b\boldsymbol{z_{l}'}$ and $\theta'=\theta_{2}$.
Moreover, $\left(\boldsymbol{\alpha'},\,\boldsymbol{\beta'},\,\boldsymbol{z'},\theta'\right)$
is a feasible dual solution to the $k$-$FTRA$ problem.\label{lem:1bound}\end{lemma}
\begin{proof}
From the constructions of $\left(\boldsymbol{\alpha_{l}'},\,\boldsymbol{\beta_{l}'},\,\boldsymbol{z_{l}'}\right)$
and $\left(\boldsymbol{\alpha_{s}'},\,\boldsymbol{\beta_{s}'},\,\boldsymbol{z_{s}'}\right)$,
we get $\forall i\in\mathcal{F}:\,\sum_{j\in\mathcal{C}}\beta_{s,ij}'\leq f_{i}+z_{s,i}'+\theta_{2}$
and $\sum_{j\in\mathcal{C}}\beta_{l,ij}'\leq f_{i}+z_{l,i}'+\theta_{1}$,
then $\sum_{j\in\mathcal{C}}\beta_{ij}'\leq f_{i}+z_{i}'+a\theta_{2}+b\theta_{1}\leq f_{i}+z_{i}'+\theta_{2}$
after multiplying the first inequality by $a$, the second by $b$
and adding them together. In addition, with the setting $\forall i\in\mathcal{F},\, j\in\mathcal{C}:\,\beta_{ij}'=\alpha_{j}'-c_{ij}$,
we get the feasibility of $\left(\boldsymbol{\alpha'},\,\boldsymbol{\beta'},\,\boldsymbol{z'},\theta'\right)$
to LP \eqref{eq:kftra-dual}. Next, we aim to derive the following
bound 
\begin{equation}
\left(2+\frac{1}{N}\right)F_{l}+C_{l}\leq\left(2+\frac{1}{N}\right)\left(\sum_{j\in\mathcal{C}}r_{j}\alpha_{l,j}'-\sum_{i\in\mathcal{F}}R_{i}z_{l,i}'-k_{l}\theta_{2}\right)
\end{equation}

from the inequality (10). For now, suppose this bound holds, from
(11), we have

\begin{equation}
\left(2+\frac{1}{N}\right)F_{s}+C_{s}\leq\left(2+\frac{1}{N}\right)\left(\sum_{j\in\mathcal{C}}r_{j}\alpha_{s,j}'-\sum_{i\in\mathcal{F}}R_{i}z_{s,i}'-k_{s}\theta_{2}\right).
\end{equation}
After multiplying (12) by $b$, (13) by $a$ and adding them together,
we get 

\begin{equation}
\left(2+\frac{1}{N}\right)F'+C'\leq\left(2+\frac{1}{N}\right)\left(\sum_{j\in\mathcal{C}}r_{j}\alpha_{j}'-\sum_{i\in\mathcal{F}}R_{i}z_{i}'-k\theta_{2}\right).
\end{equation}
This then yields the lemma together with the feasibility of $\left(\boldsymbol{\alpha'},\,\boldsymbol{\beta'},\,\boldsymbol{z'},\theta'\right)$
and $cost\left(\boldsymbol{x'},\,\boldsymbol{y'}\right)=F'+C'$. The
last thing left is to verify in the following that (12) indeed holds
from the inequality (10), the termination condition of the algorithm
$\theta_{2}-\theta_{1}\leq\epsilon=\frac{c_{min}}{8N^{2}}$ and the
fact that $C_{l}=\sum_{i\in\mathcal{F}}\sum_{j\in\mathcal{C}}c_{ij}x_{l,ij}\geq c_{min}$.

\begin{eqnarray*}
 & C_{l} & \leq2\left(\sum_{j\in\mathcal{C}}r_{j}\alpha_{l,j}'-\sum_{i\in\mathcal{F}}R_{i}z_{l,i}'-k_{l}\theta_{1}-F_{l}\right)\\
 &  & \leq2\left(\sum_{j\in\mathcal{C}}r_{j}\alpha_{l,j}'-\sum_{i\in\mathcal{F}}R_{i}z_{l,i}'-k_{l}\theta_{2}-F_{l}\right)+\frac{c_{min}k_{l}}{4N^{2}}\\
 &  & \leq2\left(\sum_{j\in\mathcal{C}}r_{j}\alpha_{l,j}'-\sum_{i\in\mathcal{F}}R_{i}z_{l,i}'-k_{l}\theta_{2}-F_{l}\right)+\frac{C_{l}k_{l}}{4N^{2}}.
\end{eqnarray*}

For simplicity, let $\triangle=\left(\sum_{j\in\mathcal{C}}r_{j}\alpha_{l,j}'-\sum_{i\in\mathcal{F}}R_{i}z_{l,i}'-k_{l}\theta_{2}-F_{l}\right)$.
Because $k_{l}\leq N$ and $N\geq1$, we get $C_{l}\leq\frac{2}{1-\frac{k_{l}}{4N^{2}}}\triangle\leq\left(2+\frac{1}{N}\right)\triangle$.
Hence, the inequality (12) is verified.
\end{proof}
For runtime, our BS procedure totally makes $O\left(L+\log N+\log n\right)$
probes ($L$ is the number of bits of the input costs) over the interval
$\left[0,\,\frac{n_{c}c_{max}}{2}\right]$ until the interval becomes
the size of $\frac{c_{min}}{8N^{2}}$. Moreover, each probe takes
$ $$O\left(n^{4}\right)$ to invoke the APD algorithm, so the total
time is $O\left(n^{4}\left(L+\log N+\log n\right)\right)$ which dominates
the overall runtime of the algorithm PK.

In the next step, we focus on analyzing the GP and RR procedures to
bound $cost\left(\boldsymbol{x},\,\boldsymbol{y}\right)=F+C$ with
$cost\left(\boldsymbol{x'},\,\boldsymbol{y'}\right)=F'+C'$ where
we denote $F$ and $C$ as the expected total facility and connection
costs respectively from the randomized procedure RR. This procedure
can also be derandomized using the method of conditional expectation
as in \cite{jain01approximation} for the $k$-median problem. In
the following, we bound $F$ with $F'$ and $C$ with $C'$ separately.
With probability 1, RR opens exactly $k$ facilities. Specifically,
it randomly opens each facility in $\boldsymbol{y_{p}}$ with probability
$b$, and each facility in $\boldsymbol{\bar{y_{p}}}$ with probability
$\frac{k-k_{s}}{k_{l}-k_{s}}$ which is also $b$. Since GP properly
splits the vector $\boldsymbol{y_{l}}$ into $\boldsymbol{y_{p}}$
and $\boldsymbol{\bar{y_{p}}}$ s.t. $\boldsymbol{y_{l}}=\boldsymbol{y_{p}}+\boldsymbol{\bar{y_{p}}}$,
we can conclude that each facility in $\boldsymbol{y_{l}}$ is opened
with probability $b$. In addition, RR randomly opens each facility
in $\boldsymbol{y_{s}}$ with probability $a$, therefore the total
expected opening cost is $aF_{s}+bF_{l}$ which is $F'$.
\begin{lemma}
The total expected facility opening cost $F$ satisfies $F=F'$.
\end{lemma}

Now we bound $C$ with $C'$. Suppose two $FTRA$ instances with the
solutions $\left(\boldsymbol{x_{s}},\,\boldsymbol{y_{s}}\right)$
and $\left(\boldsymbol{x_{l}},\,\boldsymbol{y_{l}}\right)$ are produced
from the BS procedure. Afterwards, for getting a feasible solution
to $k$-$FTRA$ from these solutions, instead we consider a naive
pseudo-polynomial time algorithm. The algorithm first treats the $FTRA$
instances with the solutions $\left(\boldsymbol{x_{s}},\,\boldsymbol{y_{s}}\right)$
and $\left(\boldsymbol{x_{l}},\,\boldsymbol{y_{l}}\right)$ as equivalent
$FTFL$ instances (by naively splitting sites and keeping the clients
unchanged) with the transformed solutions $\left(\check{\boldsymbol{x_{s}}},\,\check{\boldsymbol{y_{s}}}\right)$
and $\left(\check{\boldsymbol{x_{l}}},\,\check{\boldsymbol{y_{l}}}\right)$%
\footnote{W.l.o.g., the solutions can be easily transformed between $FTRA$
and $FTFL$ as shown in Theorem 7 of \cite{kewen2011cocoon}. %
} respectively. Then, it uses the matching and rounding procedures
\cite{Swamy08FTFL2.076} on $\left(\check{\boldsymbol{x_{s}}},\,\check{\boldsymbol{y_{s}}}\right)$
and $\left(\check{\boldsymbol{x_{l}}},\,\check{\boldsymbol{y_{l}}}\right)$
to get a feasible solution $\left(\check{\boldsymbol{x}},\,\check{\boldsymbol{y}}\right)$
to $k$-$FTFL$. Finally, the solution $\left(\check{\boldsymbol{x}},\,\check{\boldsymbol{y}}\right)$
can be easily transformed to a feasible solution $\left(\boldsymbol{x},\,\boldsymbol{y}\right)$
to $k$-$FTRA$. Now, the important observation is that directly applying
GP and RR to these $FTRA$ instances with the solutions $\left(\boldsymbol{x_{s}},\,\boldsymbol{y_{s}}\right)$
and $\left(\boldsymbol{x_{l}},\,\boldsymbol{y_{l}}\right)$ from BS
obtains essentially the \textit{same solution} $\left(\boldsymbol{x},\,\boldsymbol{y}\right)$
(also the same cost) to $k$-$FTRA$ as the naive algorithm does.
It is mainly because for the $FTRA$ instances of size $O\left(n\right)$,
our designed GP procedure pairs the integer vectors in polynomial
time. This is the acceleration of the matching procedure therein \cite{Swamy08FTFL2.076}
applied to the equivalent $FTFL$ instances of size $ $$O\left(\sum_{i\in\mathcal{F}}R_{i}\right)$.
Therefore, \textit{only in the analysis}, we can consider the naive
algorithm instead to get the following bound for $C$. This analysis
trick is similar to the trick used for analyzing the algorithm AGA
(cf. the paragraph before Theorem \ref{thm:1.52uniform}).
\begin{lemma}
The total expected connection cost $C$ satisfies $C\leq\left(1+\max\left(a,b\right)\right)C'$.\end{lemma}
\begin{proof}
For the equivalent $FTFL$ instances, we let $ $$\mathcal{F}'$ be
the set of split facilities with size $\sum_{i\in\mathcal{F}}R_{i}$
and use $k$ to index these facilities. After the matching and rounding
procedures in \cite{Swamy08FTFL2.076} on the transformed solutions
$\left(\check{\boldsymbol{x_{s}}},\,\check{\boldsymbol{y_{s}}}\right)$
and $\left(\check{\boldsymbol{x_{l}}},\,\check{\boldsymbol{y_{l}}}\right)$,
we get the solution $\left(\check{\boldsymbol{x}},\,\check{\boldsymbol{y}}\right)$
to $k$-$FTFL$. Also from its Lemma 7.2, we can directly obtain the
bound $C_{j}\leq\left(1+\max\left(a,b\right)\right)\sum_{k\in\mathcal{F}'}c_{kj}\left(a\check{x_{s,kj}}+b\check{x_{l,kj}}\right)$%
\footnote{Note that from proof of the lemma therein, we can easily get the bound
coefficient $\left(1+\max\left(a,b\right)\right)$ rather than $2$.%
} where $C_{j}=\sum_{k\in\mathcal{F}'}c_{kj}\check{x_{kj}}$, i.e.
the expected connection cost of any client $j$. Since $\left(\check{\boldsymbol{x_{s}}},\,\check{\boldsymbol{y_{s}}}\right)$,
$\left(\check{\boldsymbol{x_{l}}},\,\check{\boldsymbol{y_{l}}}\right)$
and $\left(\boldsymbol{x},\,\boldsymbol{y}\right)$ are transformed
from $\left(\boldsymbol{x_{s}},\,\boldsymbol{y_{s}}\right)$, $\left(\boldsymbol{x_{l}},\,\boldsymbol{y_{l}}\right)$
and $\left(\check{\boldsymbol{x}},\,\check{\boldsymbol{y}}\right)$
respectively with the same costs, we have
\begin{eqnarray*}
 & C & =\sum_{i\in\mathcal{F}}\sum_{j\in\mathcal{C}}c_{ij}x_{ij}=\sum_{j\in\mathcal{C}}\sum_{k\in\mathcal{F}'}c_{kj}\check{x_{kj}}\\
 &  & \leq\left(1+\max\left(a,b\right)\right)\sum_{j\in\mathcal{C}}\sum_{k\in\mathcal{F}'}c_{kj}\left(a\check{x_{s,kj}}+b\check{x_{l,kj}}\right)\\
 &  & =\left(1+\max\left(a,b\right)\right)\sum_{i\in\mathcal{F}}\sum_{j\in\mathcal{C}}c_{ij}\left(ax_{s,ij}+bx_{l,ij}\right)\\
 &  & =\left(1+\max\left(a,b\right)\right)C',
\end{eqnarray*}
which concludes the lemma.
\end{proof}
Adding up the separate bounds in the previous two lemmas, we get $cost\left(\boldsymbol{x},\,\boldsymbol{y}\right)=F+C\leq F'+\left(1+\max\left(a,b\right)\right)C'$.
Relating this bound to the bound $\left(2+\frac{1}{N}\right)F'+C'\leq\left(2+\frac{1}{N}\right)cost\left(\boldsymbol{\alpha'},\,\boldsymbol{\beta'},\,\boldsymbol{z'},\theta'\right)$
in Lemma \ref{lem:1bound}, we obtain

\begin{eqnarray*}
 & cost\left(\boldsymbol{x},\,\boldsymbol{y}\right) & \leq F'+\left(1+\max\left(a,b\right)\right)C'\\
 &  & <\left(2+\frac{1}{N}\right)\left(1+\max\left(a,b\right)\right)F'+\left(1+\max\left(a,b\right)\right)C'\\
 &  & \leq\left(2+\frac{1}{N}\right)\left(1+\max\left(a,b\right)\right)cost\left(\boldsymbol{\alpha'},\,\boldsymbol{\beta'},\,\boldsymbol{z'},\theta'\right)\\
 &  & <4cost\left(\boldsymbol{\alpha'},\,\boldsymbol{\beta'},\,\boldsymbol{z'},\theta'\right).
\end{eqnarray*}
The last inequality is from the fact that $a=\frac{k_{l}-k}{k_{l}-k_{s}}\leq1-\frac{1}{N}$
(achieved when $k_{l}=N$ and $k_{s}=k-1$), and $b=\frac{k-k_{s}}{k_{l}-k_{s}}\leq1-\frac{1}{k}$
(achieved when $k_{l}=k+1$ and $k_{s}=1$). Therefore, $1+\max\left(a,b\right)\leq2-\frac{1}{N}$
and $\left(2+\frac{1}{N}\right)\left(1+\max\left(a,b\right)\right)\leq4-\frac{1}{N^{2}}$.
By the weak duality theorem, the approximation ratio is 4. For runtime,
from the algorithm PK, both GP and RR take $O\left(n^{2}\right)$. 
\begin{theorem}
Algorithm PK is 4-approximation for the uniform $k$-$FTRA$ in polynomial
time $O\left(n^{4}\left(L+\log N+\log n\right)\right)$.
\end{theorem}

\section{Concluding Remarks}

In this paper, we studied the Constrained Fault-Tolerant Resource
Allocation ($FTRA$) problem and its important variant Constrained
Fault-Tolerant $k$-Resource Allocation ($k$-$FTRA$) problem. In
particular, although $FTRA$ generalizes the classical Fault-Tolerant
Facility Location ($FTFL$) problem, we have shown that it can achieve
the same approximation ratios as $FTFL$, for the general and the
uniform cases respectively. 

From the practical side, our developed resource allocation models
inherited from $FTFL$ and $UFL$ are more general and applicable
for optimizing the performances of many contemporary distributed systems.
Therefore, in future, it is worth looking at these models' other important
variants such as the capacitated variant in \cite{kewen2011cocoon},
the Reliable Resource Allocation ($RRA$) problem in \cite{kewen2012cats}
and etc. From the theoretical side, two grand challenges on the classical
problems still remain today: 1) close the approximation gap between
$FTFL$ (1.7245) and $UFL$ (1.488) or show $FTFL$ is more difficult
than $UFL$; 2) reduce the ratio of 1.488 to the established lower
bound 1.463 for $UFL$.

\section*{Acknowledgement} This research was partially supported by Australian Research Council Discovery Project grant \#DP0985063, National Science Foundation of China under its General Projects funding \#61170232 and Research Initiative Grant of Sun Yat-sen University.

\bibliographystyle{plain} \bibliographystyle{plain}
\bibliography{facilityLocation_new}

\end{document}